\documentclass[submission,copyright,creativecommons]{eptcs}
\providecommand{\event}{QPL 2018} 
\usepackage{breakurl} 
\usepackage{epigraph}
\usepackage{graphicx}

\usepackage[utf8]{inputenc}
\usepackage[english]{babel}
\usepackage[T1]{fontenc}

\usepackage[usenames,dvipsnames]{xcolor}
\usepackage{graphicx}

\usepackage[sans]{dsfont} 
\usepackage{bbm}
\usepackage[mathscr]{euscript}
\usepackage{cancel} 
\usepackage{amsfonts}

\usepackage{tcolorbox}
\tcbuselibrary{theorems} 

\usepackage{amsmath}
\usepackage{amssymb}
\usepackage{units}
\usepackage{framed}
\usepackage{mdframed}
\usepackage{thmtools, thm-restate} 
\usepackage{ucs}
\usepackage{subcaption}
\usepackage{csquotes}
\usepackage{epigraph}

\usepackage{enumerate}

\usepackage{tikz}
\usetikzlibrary{arrows}
\usepackage{rotating, xcolor}
\usetikzlibrary{shapes}
\usetikzlibrary{hobby}
\usetikzlibrary{decorations.markings}

\makeatletter
\def\l@subsubsection#1#2{}
\def\l@subsection#1#2{}
\makeatother

\tcbset{colback=orange!5,colframe=orange!75!yellow, 
fonttitle= \bfseries, float=htb}

    \newcommand{\ket}[1]{\vert  #1 \rangle}
    \newcommand{\bra}[1]{\langle #1 |}
    \newcommand{\inprod}[2]{\langle #1 | #2 \rangle}
	\newcommand{\proj}[2]{\ket{#1}\bra{#2}}
	\newcommand{\pure}[1]{\proj{#1}{#1}}

\newcommand{\footnoterecall}[1]{\hyperref[#1]{\footnotemark[\value{#1}]}}

\declaretheorem[]{assumption}
\declaretheorem[]{axiom}
\declaretheorem[]{definition}
\declaretheorem[sibling=definition]{theorem}
\declaretheorem[]{example}
\declaretheorem[sibling=definition]{claim}

\newenvironment{proof}{\paragraph{Proof:}}{\hfill$\square$}
\newenvironment{remark}{\textit{Remark:}}{}

\title{Inadequacy of Modal Logic in Quantum Settings}
\author{Nuriya Nurgalieva
\institute{Institute for Theoretical Physics, ETH Z\"{u}rich, 8093 Z\"{u}rich, Switzerland}
\email{nuriyan@phys.ethz.ch} 
\and
Lídia del Rio
\institute{Institute for Theoretical Physics, ETH Z\"{u}rich, 8093 Z\"{u}rich, Switzerland}
\email{lidia@phys.ethz.ch}
}
\def\titlerunning{Inadequacy of Modal Logic in Quantum Settings}
\def\authorrunning{N. Nurgalieva \& L. del Rio}

\begin{document}
\maketitle

\begin{abstract}
We test the principles of classical modal logic in fully quantum settings.
Modal logic models our reasoning in multi-agent problems, and allows us to solve puzzles like the muddy children paradox.  
The Frauchiger-Renner thought experiment highlighted fundamental problems in applying classical reasoning when quantum agents are involved; we take it as a guiding example to test the axioms of classical modal logic. 
In doing so, we find a problem in the original formulation of the Frauchiger-Renner theorem: a missing assumption about unitarity of evolution is necessary to derive a contradiction and prove the theorem. Adding this assumption clarifies how different interpretations of quantum theory fit in, i.e., which properties they violate. 
Finally, we show how most of the axioms of classical modal logic break down in quantum settings, and attempt to generalize them. Namely, we introduce constructions of trust and context, which highlight the importance of an exact structure of trust relations between agents. We propose a challenge to the community: to find conditions for the validity of trust relations, strong enough to exorcise the paradox and weak enough to still  recover classical logic.
\end{abstract}

\maketitle

\setlength{\epigraphwidth}{5in}
\epigraph{

Draco said out loud, ``I notice that I am confused.''

\textit{Your strength as a rationalist is your ability to be more confused by fiction than by reality...}

Draco was confused. 

Therefore, something he believed was fiction.} 
{Eliezer Yudkowsky, \href{https://www.lesserwrong.com/hpmor}{Harry Potter and the Methods of Rationality}}

\section{Introduction}
\label{sec:introduction}

When we talk of paradoxes in science and maths, we can usually boil the discussion down to an arising contradiction between two or more statements. These statements are necessarily associated  with the agents that deduce them, who could either be an external reader or explicit  participants in the setup of the paradox. Such agents start from shared and private prior knowledge of some facts, and along the course of the experiment apply  rules to combine these facts with what they experience, and thus arrive to a contradiction.
In other words, the agents have to apply a certain logic to conduct their reasoning. This type of classical argumentation  is illustrated in various logical puzzles --- see Example~\ref{fig:hats}. 

\begin{example}[Classical hat puzzle.]
Three wise people stand in a line. A hat is put on each of their heads. They are told that each of these hats was selected from a group of five hats: two black hats and three white hats. Arren, standing at the front of the line, can't see either of the people behind him or their hats. Tehanu, in the middle, can see only Arren and his hat. Ged can see both Arren and Tehanu and their hats.
None of the people can see the hat on their own heads. They are then asked to deduce the colour of their own hat: they are allowed to make an announcement every time a bell rings. The first time this happens, no one makes an announcement. At the second bell ring, there is again no announcement. 
Finally, at the third bell ring, Arren makes an announcement, and correctly guesses the colour of his own hat.  
What colour is his hat, and how did he come to the right conclusion? A solution using classical modal logic can be found in Appendix~\ref{appendix:modal}.
This formulation is adapted from~\cite{ThePuzzle}.
    \label{fig:hats}
\end{example}

Classical modal logic is one such system of axioms that allow us to successfully model the reasoning of agents in classical settings \cite{Lewis1918, Lewis1959, Carnap1988, Goldblatt2006, Kripke2007}. However, as we show here, it cannot be applied to more general settings where quantum experiments are conducted.
We highlight the need for more explicit specification of axioms and rules which agents are allowed to use in quantum setups, and draw parallels between the assumptions made in the Frauchiger-Renner thought experiment and the assumptions of classical modal logic.

\subsection{The Frauchiger-Renner experiment}

\begin{figure}[t]
    \centering
    \includegraphics[width=0.7\textwidth]{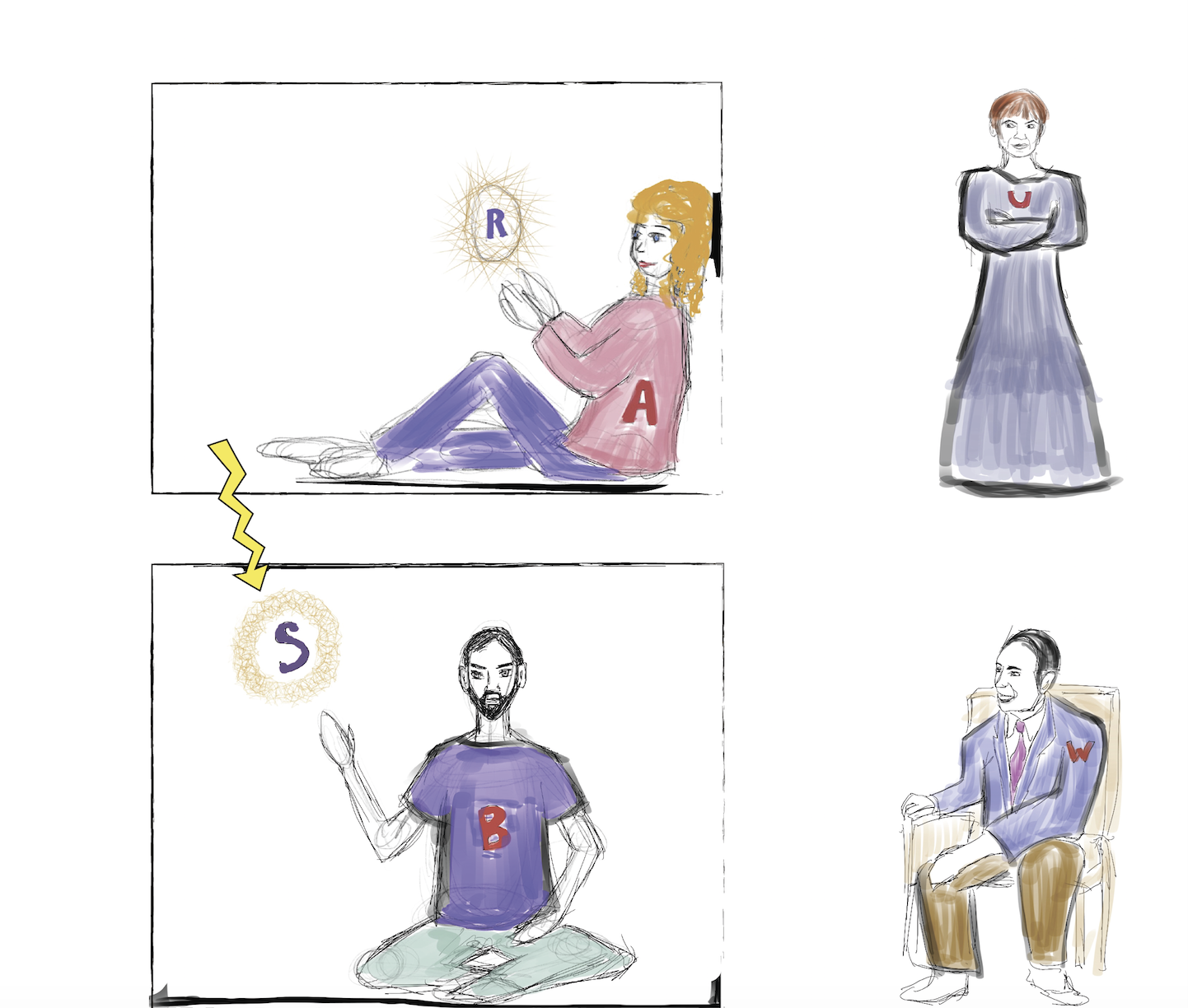}
    \caption{{\bf The Frauchiger-Renner thought experiment.} The setup of the experiment \cite{Frauchiger2018}. Alice measures a qubit $R$ and, depending on the outcome, prepares a qubit $S$ and sends it to Bob, who measures the spin; Alice and Bob's labs are measured by Ursula and Wigner respectively.}
    \label{fig:frauchigerrennersetup}
\end{figure}

We focus on settings where agents are equipped with quantum memories, and 
must argue explicitly about the outcomes of measurements performed by themselves and other agents on such memories. 
A good testing ground is given by Daniela Frauchiger and Renato Renner's  extension of the Wigner's friend thought experiment \cite{Frauchiger2018}. The experiment is schematically shown in Figure~\ref{fig:frauchigerrennersetup}. 

\begin{example}[Frauchiger-Renner thought experiment \cite{Frauchiger2018}.]
This example consists of four participants, Alice, Bob, Ursula and Wigner. Each experimenter is equipped with a quantum memory ($A, B, U$ and $W$, respectively).
In addition, there are two other systems, $R$ and $S$, which are qubits (Figure~\ref{fig:frauchigerrennersetup}).

\paragraph{Initial setting.}
The initial state of the $R$ is $\sqrt{\frac{1}{3}}\ket{0}_R+\sqrt{\frac{2}{3}}\ket{1}_R$. The initial state of $S$ is $\ket{0}_S$. The initial state of the relevant subsystems of the agents' memories is $\ket0_A, \ket0_B, \ket0_U $ and $\ket0_W$. 

\paragraph{Proceeding.} The agents proceed as follows:
\begin{enumerate}[{$t= $} 1.]
    \item Alice measures system $R$ in basis $\{\ket0_R, \ket1_R\}$. She records the result in her memory $A$, and prepares $S$ accordingly: if she obtains outcome $a=0$ she keeps her memory in state $\ket0_A$ and $S$ in state $\ket0_S$; if she obtains outcome $a=1$ she changes her memory to  $\ket1_A$ and $S$ to $\frac1{\sqrt2}(\ket0_S+\ket1_S)$.
    Finally, she gives system $S$ to Bob. 

    \item Bob measures system $S$ in basis $\{\ket0_S, \ket1_S\}$ and records the outcome $b$ in his memory $B$, similarly to Alice. 
    
    \item Ursula measures Alice's lab (consisting of $R$ and the memory $A$) in basis $\{\ket{ok}_{RA}, \ket{fail}_{RA}\}$, where
\begin{gather*}
\ket{ok}_{RA}=\sqrt{\frac{1}{2}}(\ket{0}_R\ket{0}_A-\ket{1}_R\ket{1}_A)\\
\ket{fail}_{RA}=\sqrt{\frac{1}{2}}(\ket{0}_R\ket{0}_A+\ket{1}_R\ket{1}_A).
\end{gather*}

    \item Wigner measures Bob's lab (consisting of $S$ and the memory $B$) in basis $\{\ket{ok}_{SB}, \ket{fail}_{SB}\}$, where
\begin{gather*}
\ket{ok}_{SB}=\sqrt{\frac{1}{2}}(\ket{0}_S\ket{0}_B-\ket{1}_S\ket{1}_B)\\
\ket{fail}_{SB}=\sqrt{\frac{1}{2}}(\ket{0}_S\ket{0}_B+\ket{1}_S\ket{1}_B).
\end{gather*}

    \item Ursula and Wigner compare the outcomes of their measurements. If they were both ``ok'', they halt the experiment. Otherwise, they reset the timer and all systems to the initial conditions, and repeat the experiment. 
\end{enumerate}

Does the experiment ever halt? In that case, what can Wigner conclude about the reasoning of all the other agents? 

\end{example}

It can be shown that if we follow the Born rule, this experiment will at some point halt, as the overlap of the global state at time $t=2.5$ with $\ket{ok}_{RA}\ket{ok}_{SB}$ gives us a probability of $1/12$.
If one applies the intuitive rules of classical reasoning to all agents, then we reach a contradiction:  Wigner deduces that whenever the experiment halts, Alice, at time $t=1$, had predicted with certainty that it would not.  Let us see how. 

\paragraph{Intuition behind the paradox.}
Here we sketch of the reasoning of agents leading to the paradox in  the Frauchiger-Renner thought experiment. For now we will use standard quantum theory, as seen from the outside, and naive logical inference; the formalisation and criticisms to this approach will come later. 

\begin{enumerate}[{$t= $} 1.]
    \item After Alice measures $R$  and records the result, the joint state of $R$ and her memory $A$ becomes entangled,
    \begin{gather*}
    \left(\sqrt{\frac{1}{3}}\ket{0}_R+\sqrt{\frac{2}{3}}\ket{1}_R\right)\ket{0}_A 
    \quad \longrightarrow \quad
\sqrt{\frac{1}{3}}\ket{0}_R\ket{0}_A+\sqrt{\frac{2}{3}}\ket{1}_R\ket{1}_A.
\end{gather*}
    When Alice prepares $S$, it too becomes entangled with $R$ and $A$,
     \begin{gather*}
\sqrt{\frac{1}{3}}\ket{0}_R\ket{0}_A \ket0_S+\sqrt{\frac{2}{3}}\ket{1}_R\ket{1}_A \  \frac1{\sqrt2}(\ket0_S +\ket1_S).
\end{gather*}

    \item When Bob measures $S$ and writes the result in his memory, the global state becomes
    \begin{gather*}
\ket{\psi}_{RASB} = \frac{1}{\sqrt{3}}\left(\ket{0}_R\ket{0}_A \ket0_S\ket0_B+ \ket{1}_R\ket{1}_A\ket0_S\ket0_B +\ket{1}_R\ket{1}_A\ket1_S\ket1_B \right).
\end{gather*}
    Crucially, this is a Hardy state \cite{Hardy1993}, and its terms  can be rearranged into two  convenient forms, which will be used later,
    \begin{align*}
\ket{\psi}_{RASB} &= \sqrt{\frac{2}{3}} \underbrace{\frac1{\sqrt2}\left(\ket{0}_R\ket{0}_A + \ket{1}_R\ket{1}_A \right)}_{\ket{fail}_{RA}}\ket0_S\ket0_B +
\frac{1}{\sqrt{3}}\ket{1}_R\ket{1}_A\ket1_S\ket1_B \\
&= \frac{1}{\sqrt{3}}\ket{0}_R\ket{0}_A \ket0_S\ket0_B + \sqrt{\frac{2}{3}} \ket1_R\ket1_A \underbrace{\frac1{\sqrt2}\left(\ket{0}_S\ket{0}_B + \ket{1}_S\ket{1}_B \right)}_{\ket{fail}_{SB}} .
\end{align*}
   
\item 
Now Ursula and Wigner measure Alice's and Bob's labs in their bases listed above; the possibility of them both getting the outcome ``ok'' is non-zero:
\begin{gather*}
P[u=w=ok]=|(\bra{ok}_{RA}\bra{ok}_{SB})\ket{\psi}_{RASB}|^2=\frac{1}{12}.
\end{gather*}
From now on, we post-select on this event. 
At time $t=3$, Ursula reasons about the outcome that Bob observed at $t=2$. Since 
$$\ket{\psi}_{RASB}  = \sqrt{\frac23} \ket{fail}_{RA}\ket0_S \ket0_B + \frac{1}{\sqrt{3}}\ket{1}_R\ket{1}_A\ket1_S\ket1_B,$$ 
she concludes that the only possibility with non-zero overlap with her observation of $\ket{ok}_{RA}$ is that Bob measured $\ket{1}_S$.  We can write this as ``$u=ok \implies b=1$''. 
She can further reason about what Bob, at time $t=2$ thought about Alice's outcome at time $t=1$. 
Whenever Bob observes $\ket1_S$, he can use the  same form of $\ket{\psi}_{RASB}$ to conclude that Alice must have measured $\ket1_R$. We can write this as ``$b=1 \implies a=1$''.
Finally, we can think about Alice's deduction about Wigner's outcome. Using the form 
$$\ket{\psi}_{RASB} =  \frac{1}{\sqrt{3}}\ket{0}_R\ket{0}_A \ket0_S\ket0_B + \sqrt{\frac{2}{3}} \ket1_R\ket1_A \ket{fail}_{SB},$$
we see that Alice reasons that, whenever she finds $R$ in state $\ket1_R$, then Wigner will obtain outcome ``fail'' when he measures Bob's lab. That is, ``$a=1 \implies w=fail$''.

\end{enumerate}

Thus, chaining together the statements (the same reasoning that allowed the reader to solve the three hats problem), we reach an apparent contradiction: 
$$w=u=ok \implies b=1\implies a=1\implies w=fail.$$ 
That is, 
when the experiments stops with $u=w=ok$, the agents  can make \emph{deterministic}  statements about each other's reasoning and measurement results, concluding that Alice had predicted $w=fail$. 

We notice that we are confused. It follows that some of the assumptions that we implicitly used to derive the contradiction must be inadequate in this setting --- or simply incompatible with each other. In the next section we will first review the original attempt at listing these assumptions, and then we will show that it was still incomplete.

\section{Analysis of assumptions in the Frauchiger-Renner result}
\label{sec:reformulation}

In their original article, Frauchiger and Renner \cite{Frauchiger2018} use the above thought experiment to formulate a no-go theorem under three assumptions. These assumptions correspond in spirit to: compatibility with the Born rule of quantum theory \textbf{Q}, logical consistency among agents \textbf{C}, and experimenters having the subjective experience of only seeing one outcome \textbf{S}.
The paradox is formulated in terms of statements different agents make using the assumptions. It is assumed that all  agents employ the same theory $T$ for reasoning, and the assumptions capture the essence of this theory. We present them here as stated in \cite{Frauchiger2018}.\footnote{The exact formalization of the assumptions is not too important for the conceptual discussion of this work, and it suffices for the reader to understand the intuition behind them.}

The first assumption encapsulates the idea of the observers in the experiment ``using quantum theory''; to be precise, they are using not the ``full version'' of quantum mechanics, but rather a weak version of the Born rule applicable only to deterministic situations.
\begin{assumption}[Q \cite{Frauchiger2018}]
A theory $T$ that satisfies \textbf{Q} allows any agent Alice to reason as follows.
Let $S$ be an arbitrary system around Alice, $\psi$ a unit vector of its Hilbert space, and $\{\pi_x^H\}_{x \in X}$ a family of positive operators on this space such that  $\sum_x \pi_x^H=id$, and  $\bra{\psi} \pi_\xi^H  \ket{\psi}=1$ for some $\xi\in X$. 
Suppose that Alice has established the statements
``System $S$ is in state $\psi$ at time $t_0$.'' and  ``The value $x$ is obtained by a measurement of $S$ w.r.t. the family $\{\pi_x^H\}$ (in the Heisenberg picture, relative to time $t_0$). The measurement is complete at time $t$.''  Then Alice can conclude ``I am certain at time $t_0$ that $x=\xi$ at time $t$.''
\end{assumption}

The second assumption governs how agents  reason from the viewpoint of each other.
\begin{assumption}[C \cite{Frauchiger2018}]
A theory $T$ that satisfies \textbf{C} allows any agent Alice to reason as follows. If Alice has established  ``I am certain at time $t_0$ that agent $B$, upon reasoning using $T$, is certain that $x=\xi$ at time $t$'' where $x$ is a value that can be observed at time $t > t_0$, then Alice can conclude
``I am certain at time $t_0$ that $x=\xi$ at time $t$.''
\end{assumption}

The third assumption captures the idea of each agent experiencing only single outcome of a measurement. For example, if Alice is going to measure $R$, she is not allowed to state both ``I am certain that I will observe outcome 0'' and ``I am certain that I will observe outcome 1.''
In the language of the many-worlds interpretation, this corresponds to associating the agent's experience with a single ``branch'' of the global wave function (in the measurement basis), and not with the superposition of the different branches.
\begin{assumption}[S \cite{Frauchiger2018}]
A theory $T$ satisfies \textbf{S} if it disallows any agent Alice to make both the statements ``I am certain at time $t_0$ that $x=\xi$ at time $t$.'' and
``I am certain at time $t_0$ that $x\neq\xi$ at time $t$'' where $x$ is a value that can be observed at time $t > t_0$. 
\end{assumption}

The claim by Frauchiger and Renner is then that the three assumptions are incompatible, and the proposed proof goes via an analysis of the thought-experiment. 

\begin{claim}[Frauchiger-Renner \cite{Frauchiger2018}]
No physical theory $T$ can simultaneously satisfy the assumptions \textbf{Q}, \textbf{C} and \textbf{S}.
\end{claim}

In this work, we note that there is an additional implicit assumption used in the proof \cite{Frauchiger2018}: all agents are considering the evolution of another agents in their labs unitary. In other words, the observers, when reasoning about statements other observers make, do it according to an assumption of a specific type of evolution happening in the labs. This assumption cannot, to the best of our knowledge, be derived from the other three. Indeed, the proof in the original paper uses the fact that  agents describe the evolution of each other's labs after local measurements through isometries. For example, unitarity is necessary for Ursula  to derive even the first implication ``$u=ok\implies b=1$,'' as well as for Alice to derive ``$a=1 \implies w=fail$.'' We will illustrate this with a counter-example, but before we proceed, allow us to lay out the missing assumption.

\begin{assumption}[U]
A theory $T$ that satisfies \textbf{U}
allows any agent $A$ to model measurements performed by any other agent $B$ as reversible evolutions in $B$'s lab --- for example, a unitary evolution $U_{BS}$ of the joint state of $B$'s memory and the system $S$ measured.
\end{assumption}

The exact formalisation of \textbf{U} is not important for now; let us instead see why we need it.
Suppose for example that Alice obtains the outcome $a=1$ and describes the measurement process in Bob's lab as completely positive trace-preserving map. 
She could model the global state before Bob's measurement as
$$\ket{1}_R\ket{1}_A \frac{1}{\sqrt2} (\ket0_S + \ket1_S) \ket0_B,$$
and after Bob's measurement, as a density matrix
\begin{align*}
    \rho_{RASB} &= \pure{1}_R \otimes \pure1_A \otimes \left(\frac12 \pure0_S \otimes \pure0_B + \frac12 \pure1_S \otimes \pure1_B \right) \\
    &= \pure{1}_R \otimes \pure1_A \otimes \left(\frac12 \pure{ok}_{SB} + \frac12 \pure{fail}_{SB}\right)
\end{align*}
Consequently, the probability of Wigner getting the outcome ``$w=ok$'' would be, from Alice's perspective,
\begin{gather*}
P[ w=ok \vert a=1]=tr[\rho_{SB}  \ \pure{ok}_{SB}]=\frac{1}{2}.
\end{gather*}
Thus, in this case the same reasoning used in the proof of the original article is not applicable, which highlights the importance of the implicit assumption of a specific type of evolution happening in the labs.
We may now complete the initial claim of contradiction as follows. 
The proof of the original paper~\cite{Frauchiger2018} applies. 

\begin{claim}[Reformulating the Frauchiger-Renner theorem]
No physical theory $T$ can govern the reasoning of all agents while simultaneously satisfying the assumptions \textbf{Q}, \textbf{C},  \textbf{S} and \textbf{U}. 
\end{claim}

One way to contradict assumption  \textbf{U} is by saying that unitary quantum mechanics does not apply in the regime of large bodies; this is the case of fundamental collapse theories \cite{Ghirardi1985, Ghirardi1986, Bassi2003}, for example. As these theories give different predictions to quantum mechanics, they are falsifiable. Alternatively, one could postulate that the evolution of large systems is unitary, but that there is fundamentally no way to measure them --- which is in principle falsifiable as the precision of  mesoscopic experiments improves. 
We will see further ways around \textbf U later on.
Unashamedly pedantic researchers such as the authors could point out other assumptions that are implicitly included in the common knowledge of all agents involved, for example:
\begin{itemize}
    \item[{\bf P.}]  It is possible for agents to prepare systems in \emph{pure states}, and to measure them according to perfect projectors. It is also possible to perform these measurements in isolation. 
    
    If we drop this assumption, and claim that realistically, agents can only prepare approximate states and perform noisy measurements, we don't obtain a deterministic contradiction. Nevertheless, we would still obtain a discrepancy in the predictions made by the different agents in this thought experiment.
    That is, we would still obtain a conceptual problem, even if it were not as neatly formalized as in the Frauchiger-Renner theorem. In order to formalize the contradiction, we would have to relax the version of the Born rule used here to cover frequentist or probabilistic cases. 
    We could also consider explicit reference frames for the measurements; again, the conceptual contradiction doesn't disappear.

    \item[{\bf M.}] Agent's \emph{memories} are ultimately physical systems. In particular, they are quantum systems, and different classical statements are encoded as orthogonal quantum states. (For example, Alice's statements ``I saw outcome 0'' and ``I saw outcome 1'' can be stored in her memory $A$ as states $\ket0_A$ and $\ket1_A$ respectively, with $\inprod01_A =0$.)  It is in principle possible to perform quantum measurements on these memories. 
    
    This is the approach of physicalism (``mind is matter'', most eloquently described in \href{https://youtu.be/_9Rs61l8MyY?t=35m55s}{this talk}); alternatively 
    we could imagine an interpretation of quantum theory that demands  that agents' thoughts cannot be encoded in physical systems, thus preventing us from meaningfully measuring observers.\footnote{If the reader finds this a satisfactory way out, we have \href{https://en.wikipedia.org/wiki/Ghost_in_the_machine}{different tastes} and further discussion would not be productive.}
    Even with such a theory, the unitarity assumption \textbf U may suffice to derive a contradiction (after adjusting the definition of subsystems $A$ and $B$ measured by Ursula and Wigner), so physicalism may not be necessary to find the contradiction. It may be interesting to investigate whether \textbf U implies a weak version of physicalism.

    \item[{\bf A.}]  Each agent considers the perspective of other \emph{agents} as valid as their own. 
    
    In fact, to derive the contradiction it is not strictly needed that agents in the Frauchiger-Renner experiment think of all the other agents as capable of reasoning. For example, Ursula may see Alice as simply a quantum system that she measures. She will however take Bob seriously as an agent, and therefore trust his reasoning. Ursula further knows that Bob sees Alice as an agent. Since Bob trusts Alice's reasoning, he will draw conclusions from it (``$b=1 \implies a=1 \implies w= fail$''). Ursula then only needs to trust Bob's conclusion, regardless of his source (``$u=ok \implies b=1 \implies w= fail$''). In that sense, agents don't need to directly assign agency to those whom they are measuring. 
    
    This weaker version of \textbf A roughly corresponds to the consistency assumption \textbf C. One could debate whether it is this specific assumption that is violated by agent-centric approaches to quantum theory such as QBism \cite{Caves2002, Fuchs2014, Fuchs2016}. 
    
    \item[{\bf X.}] There are no \emph{extraordinary} interventions on an experimental setting such as the one used in the proof (once all agents agree on it).
    
    This assumption is certainly needed in order to derive the paradox (otherwise the agents could not even describe the setting). However, given that most interpretations of quantum mechanics do not forbid \textbf X, it is not a very helpful assumption to help us categorize different versions of quantum theory. The same holds for assumption \textbf P above.

\end{itemize}


\subsection{Interpretations and reactions}
\label{sec:interpretations}
 Having identified the missing assumption needed to derive a contradiction, it is possible to classify theories and interpretations of quantum mechanics in terms of which assumptions they violate. This allows us to clarify some of the analysis of the original paper~\cite{Frauchiger2018}, in particular for the case of theories that violate \textbf U.

Before we proceed, we note that just because we found four assumptions that allow us to derive a contradiction, it does not mean that this is a useful partition of properties of physical and logical theories. That is, some theories and interpretations will not fit neatly into the boxes we drew for them --- we may find for example that they violate more than one assumption. As an analogy, there are many ways to split a large system into subsystems, and some will be more operational than others. Similarly, there are several ways to split the contradiction into modular assumptions, and while the current classification may appear intuitive, it is not necessarily the best in the long run. 

The list of theories discussed here is not exhaustive, and we focus on reactions to the original paper. 
For the interested reader, the theories below   are introduced in more detail in Appendix~\ref{appendix:interpret}. Further interpretations are discussed in \cite{Frauchiger2018}.

\paragraph{Bohmian theory.}

In \cite{Sudbery2016} Anthony Sudbery applies Bohmian theory to the Frauchiger-Renner experiment.
The author points out that the theory satisfies \textbf Q, \textbf C and \textbf S. 
He then argues that the analysis of the experiment according to the proof of Frauchiger and Renner assumes a particular type of evolution that does not reflect Bohmian theory. 
In Bohmian theory we associate nature with two states: {\bf(1)} the \emph{pilot vector} $\ket{\Psi(t)}$, which corresponds roughly to the global quantum state as seen from the outside, and which evolves unitarily (for example, $\ket{\Psi(t)} = \frac1{\sqrt2} (\ket0_R\ket0_A + \ket1_R\ket1_A))$; and {\bf(2)} the \emph{real state vector} $\ket{\Phi(t)}$,   which corresponds to the ``collapsed'' state after a measurement outcome, and to the experiences of the agents considered (for example $\ket{\Phi(t)} = \ket1_R\ket1_A$). Crucially, the real state need not evolve unitarily: $\ket{\Phi(t)}$ can be picked at every time instant from the pilot $\ket{\Psi(t)}$, more independently from $\ket{\Phi(t-1)}$ than in collapse theories (for  consistency rules between the real vectors over time see \cite{Sudbery2016}).  Even if Alice sees outcome $1$, corresponding to the real state $\ket{\Phi(t=1)}=\ket{1}_R\ket{1}_A \frac{1}{\sqrt2} (\ket0_S + \ket1_S) \ket0_B $, she also keeps track of the global pilot state, and evolves the latter. From the pilot  $\ket{\Psi(t=4)}$ she can extract all the possible compatible real states, of which one, $\ket{\Phi(t=4)}$, allows her to see $a=1$ while Ursula and Wigner both get $ok$; therefore, Alice would not make the prediction ``$a=1 \implies w=fail$'' (for the calculation, see \cite{Sudbery2016}). Note that if Alice had only evolved the real state without consideration for the pilot, or if she had imposed more consistency over time for the real vector, she would not be able to see this. 
We conclude that Bohmian mechanics violates assumption \textbf U, and this example can help us work out a precise formalization for \textbf U. Incidentally, this analysis suggests that pre- and  post-selection have little consequence in Bohmian theory, and it would be interesting to investigate the consequences of this.

\paragraph{Copenhagen interpretations.}
Jeffrey Bub, in his recent work \emph{Why Bohr was (mostly) right} \cite{Bub2017}, gives an interesting insight on the paradox based on Bohr's work. He operates with a notion of a Boolean frame, which is simply a classical frame, determined by the position of the Heisenberg cut. He states:
\begin{displayquote}
Once you decide what counts as the ultimate measuring instrument in a given scenario, the cut is no longer movable. If you move it, you are ``subdividing the phenomenon,'' as Bohr puts it, and the two analyses will be incompatible --- in general they won't both fit into one Boolean frame. In effect, a phenomenon is an episode from the perspective of a Boolean frame.
\end{displayquote}

Bub then argues that only the single-observer perspective is useful in quantum mechanics, and if one takes a point of view of a certain observer in the Frauchiger-Renner experiment, the measurement outcome events for all other observers are neglected; thus, no contradiction arises, as assumption \textbf C is not satisfied. 
In other words, since agents move the cut individually, one has to choose and fix a single  observer in the experiment, who makes the ultimate decision about the cut, and whose perspective is the only valid one. 
This is an example of reducing the scope of a theory until there are no mathematical contradictions left.\footnote{As physicists, we would rather explore apparent inconsistencies, and enrich the scope of the theory.}

\paragraph{QBism.}
QBism \cite{Caves2002, Fuchs2014, Fuchs2016} also restricts itself to a single-agent theory; it deals with experiences and actions of one specific agent, and thus it is not clear if the thought experiment can be analyzed within QBism without  further generalizing the theory. For the further discussion we refer to the Frauchiger-Renner paper \cite{Frauchiger2018}.

\paragraph{Unitary theories: many-worlds and relative state.}
In many-worlds interpretations \cite{Everett1957, ManyworldsStanford}, assumption \textbf U is a part of the theory (as the state of the universe evolves unitarily), while  assumptions \textbf S and \textbf C can be questioned. If the theory is formulated in a way that models agents to be able to experience more than one outcome at once, then the assumption  \textbf S is violated. 
Additionally, if the agents do not have reliable memories, then the consistency assumption \textbf C might not hold. While applying assumption \textbf C, agents put their trust in the results of other agents; this may not hold, for example, in versions of many-worlds where the observer is not associated with a particular branch.

In the article \emph{On formalisms and interpretations} \cite{Baumann2018}, Veronika Baumann and Stefan Wolf already point out that different models for the global evolution lead to different predictions in experiments where agents are measured. The authors provide  an extension of the Born rule within a unitary ``relative-state formalism,''  to compute probabilities of outcomes in settings with multiple agents  (summary in Appendix~\ref{appendix:interpret}). Finally, they show that if agents model the evolution of each other through collapse theories, the contradiction does not arise.

\section{Classical modal logic}
\label{sec:classical}
In this section we review classical modal logic, which allows us to model and solve complex scenarios where several agents think about each other's reasoning (such as the three hats puzzle). 
We then discuss which  axioms of classical modal logic can be applied to fully quantum settings, and how they relate to the assumptions in the Frauchiger-Renner result.
For conciseness, we stick to the bare bones of the modal logic framework here \cite{Lewis1918, Lewis1959, Carnap1988, Goldblatt2006, Kripke2007, Fagin2004}. For more details on motivation, development and related work, see Appendix~\ref{appendix:modal}. Note: the branch of modal logic applied to reasoning about knowledge is sometimes called epistemic logic, and modal logic in multiple agent settings is sometimes called multi-modal logic.

\paragraph{Modal logic syntax.}
Modal logic uses a number of standard logical operators (also known as connectives), such as $\neg$ for `not', $\Rightarrow$ for `if...then', $\wedge$ for `and', $\vee$ for `or' and $\Leftrightarrow$ for `equivalent.'
So how do we apply this syntax to the statements that agents make? Suppose we have a group of $n$ agents ($1,2,...,n$), who are able to describe the world in terms of primitive propositions $\phi_1,\phi_2,...$ belonging to a set $\Phi$. The primitive propositions are simple facts about the world, for example, $\phi_1=$``The dragon lives under the Lonely Mountain'' or $\phi_2=$ ``Minas Tirith is white.''
Then to express the statement like ``Bilbo knows the dragon lives under the Lonely Mountain,'' we introduce modal operators $K_1,K_2,...,K_n$, one for each agent; in that case, $K_i\phi$ corresponds to ``Agent $i$ knows $\phi$.'' 
Thus, the statements produced by the agents, including those concerning the other agents' views, can be written in a compact way.
This is especially convenient when the statements are linguistically complex: for example, the statement ``Aragorn knows that Bilbo doesn't know that Elrond knows that the dragon lives under the Lonely Mountain''  can compressed into $K_A\neg K_B K_E\phi_1$.

\paragraph{Possible worlds semantics and Kripke structures.}
In modal logic, a set $\Sigma$ of possible worlds is introduced. The truth value of a proposition $\phi$ is then assigned depending on the possible world in $\Sigma$, and can differ from one possible world to another.
We will first provide a structure which serves as a complete picture of the setup the agents are in, and then discuss the elements of the structure. 
\begin{definition}{\textbf{(Kripke structure)}}
A Kripke structure $M$ for $n$ agents over a set of statements $\Phi$ is a tuple $\langle \Sigma, \pi,\mathcal{K}_1,...,\mathcal{K}_n\rangle$ where $\Sigma$ is a non-empty set of states, or possible worlds, $\pi$ is an interpretation, and $\mathcal{K}_i$ is a binary relation on $\Sigma$.

The interpretation $\pi$ is a map $\pi:\Sigma\times\Phi\to\{\textbf{true},\textbf{false}\}$, which defines a truth value of a statement $\phi\in\Phi$ in a possible world $s\in\Sigma$.

$\mathcal{K}_i$ is a binary equivalence relation on a set of states $\Sigma$, where $(s,t)\in\mathcal{K}_i$ if agent $i$ considers world $t$ possible given his information in the world $s$. 
\end{definition}

The truth assignment  tells  us if the proposition $\phi\in\Phi$ is true or false in a possible world $s\in\Sigma$; for example, if $\phi=$ ``Minas Tirith is white,'' and $s$ is a world where the Middle-earth is real and Minas Tirith still stands, then $\pi(s,\phi)=\textbf{true}$. The truth value of a statement in a given structure $M$ might vary from one possible world to another; we will denote that $\phi$ is true in world $s$ of a structure $M$ by $(M,s)\models\phi$, and $\models\phi$ will mean that $\phi$ is true in any world $s$ of a structure $M$.

Agents do not possess the complete information about a possible world they are in, and may consider other possible worlds possible; for example, if the agent doesn't know if it is raining in the Shire, she can consider both the world where it is indeed raining in Shire, and the world where it doesn't, possible. This situation is captured by binary relations $\mathcal{K}_i$.
Formally, we say that agent $i$ ``knows'' $\phi$ in a world $s$, that is $(M,s) \models K_i \phi $, if and only if for all possible worlds $t$ such that $(s,t) \in \mathcal K_i$ (that is, all the worlds admitted by the agent given their knowledge), it holds that $(M,t) \models \phi$. 

\paragraph{Axioms of knowledge.}
In order to operate the statements agents produce, we have to establish certain rules which are used to compress or judge the statements. These are the axioms of knowledge \cite{LogicStanford}. They might seem trivial in the light of our everyday reasoning, yet given our awareness of the quantum case, we will treat them carefully.

\begin{axiom}[Distribution axiom.]
If an agent is aware of a fact $\phi$ and that a fact $\psi$ follows from $\phi$, then the agent can conclude that $\psi$ holds:
$$(M,s)\models(K_i\phi\wedge K_i(\phi\Rightarrow\psi))\Rightarrow (M,s)\models K_i\psi.$$
\end{axiom}
\begin{remark}
A reasonable implication of the distribution axiom would be that if an agent is aware that a fact $\phi$ follows from $\chi$ and $\psi$ follows from $\phi$, then he can conclude that $\psi$ follows from $\chi$:
$$(M,s)\models(K_i(\chi\Rightarrow\phi)\wedge K_i(\phi\Rightarrow\psi))\Rightarrow (M,s)\models K_i(\chi\Rightarrow\psi).$$
\end{remark}
Note that this is weaker than saying that agents know all logical consequences of their knowledge. 
For example,  it follows from the axiom that ``if Bilbo knows that it is raining in the Shire, and he knows that grass grows when it rains, then he knows that grass is growing in Shire.'' An agent less versed in botany would not be able to reach the same conclusion. 

\begin{axiom}[Knowledge generalization rule.]
All agents know all the propositions that are  valid in a structure:
\begin{center}
if $(M,s)\models\phi \ \forall s$ then $\models K_i\phi \ \forall i.$
\end{center}
\end{axiom}
An example would be ``if it always rains in the Shire, then all agents know that it rains in the Shire.''
The above appears to be a rather strong assumption, but it can be seen as a condition on what is common knowledge in a given Kripke structure.  That is, if we do not want ``rainy Shire'' to be common knowledge among a set of agents, we just need to add one possibly world $s$ where it is sunny in the Shire to the Kripke structure that models the agents and their knowledge. 
Examples of facts usually included as common knowledge in typical Kripke structures are physical laws, like ``objects fall'' or ``one day follows the other.'' We will see that this can include ``quantum theory holds.'' This axiom is sometimes called the necessitation rule. 

\begin{axiom}[Truth axiom.]
If an agent knows a fact then the fact is true,
$$(M,s)\models K_i\phi\Rightarrow (M,s)\models\phi.$$
\end{axiom}
The truth axiom is taken by many philosophers to be the major property distinguishing knowledge from belief.
It follows from the previous two axioms together that if an agents knows $\phi$ to hold in all possible worlds, then all agents know $\phi$ to be true. Not surprisingly, this will be problematic in quantum settings, where measurement outcomes and conclusions can be subjective. We will discuss this in the next section.

\begin{axiom}[Positive and negative introspection axioms.] 
Agents can perform  introspection regarding their knowledge:
\begin{center}
$(M,s)\models K_i\phi\Rightarrow (M,s)\models K_iK_i\phi$ (Positive Introspection),\\
$(M,s)\models \neg K_i\phi\Rightarrow (M,s)\models K_i\neg K_i\phi$ (Negative Introspection).
\end{center}
\end{axiom}
In other words, agents are aware of the limitations of their own knowledge. If I know $\phi$, then I know that I know $\phi$; if I don't know $\phi$, I am aware of not knowing $\phi$.  This uncompromising rationality is lacking in many agents in the real world: take for example $\phi=$``god exists'', ``this policy has good long-term effects'' or ``my friend is trustworthy.'' In the quantum case, it will be tricky for different reasons.

\subsection{Problems in quantum settings}

\label{sec:quantumproblems}

Let us now try to apply the semantics of Kripke structures to fully quantum settings, and discuss the problems that arise.

\paragraph{Scope of ``agent.''}
The largest problem that affects the consistency of the logic axioms in quantum settings is the very definition of an agent, which is taken as an abstract primitive in classical modal logic. Quantum theory forces us to associate agents with a physical (and therefore quantum) memory, as opposed to an abstract ``perspective'' and reasoning independent of a physical substrate. 
Then, for example (in a unitary  version of quantum theory), after Alice measures $R$, do we associate her with her subjective experience of a single outcome, or with her quantum memory $A$, which is entangled with $R$ and cannot be mapped to a single outcome? 
More generally, what are the constraints that we can impose on the evolution of that memory in order to say that they are the same agent? For example, most people would agree that their physical brains and bodies form their perception of self, and
if their memory was completely erased and reprogrammed, they would not be the same person. In our thought experiment, after Ursula measures Alice, thus significantly altering the state of her memory $A$, is it fair to say that pre-measurement Alice is the same agent as post-measurement Alice? We treat this issue in Section~\ref{sec:suggestions}. 

\paragraph{Possible worlds semantics.} In classical Kripke structures, we have possible worlds or states, in which the truth value of propositions can be evaluated. It is ambiguous how to adapt this semantics of ``possible worlds'' to settings of multiple quantum agents, and unclear if one can do it at all without running into contradictions. One naive approach (which we use at the end of this section to reproduce the Frauchiger-Renner paradox) is to associate each possible world with one set of outcomes of the final measurements of an experiment, that is those outcomes that can be considered ``classical'' for practical purposes, because the memories in which they are recorded will not be further modified. In our example, this will correspond to a possible world for each pair of outcomes $(u,w)$ of Ursula and Wigner. This set of possible worlds will allow us to assess the truth value of the propositions needed for the paradox, but it does not let us determine the truth of arbitrary statements, like ``if Alice measured in a different basis, she would obtain outcome 0.'' These statements may indeed be undefined in the Kripke structure that describes the experiment. 
Our intuition is that the very notions of possible worlds and truth assignments are too classical in nature, and that this is where we need to start to change the modal logic framework in order to generalize it to quantum settings.

\paragraph{Common knowledge and physical laws.}
As we saw, the knowledge generalization rule is a constraint on what is common knowledge, and therefore should apply to the general statements  that all agents can make about the setup and the course of the experiment.
In the Frauchiger-Renner setting, this axiom covers the assumptions \textbf{Q} and \textbf{U}, which are taken as physical laws accepted by all agents. Note however that if two agents use different versions of quantum theory (for example Alice uses an interpretation with fundamental collapse whereas Bob follows unitary quantum theory), then  assumptions like  \textbf{U} may not be common knowledge, and may therefore not be always true in the Kripke structure modelling those agents.  

\paragraph{Introspection, single outcome and context.}
\label{par:intr&cont}
The introspection axioms are not directly applied in the Frauchiger-Renner setting, but they do  resemble assumption \textbf{S}: if an agent can only perceive one outcome of a measurement, she is able to give a certain answer when asked about the outcome, which essentially means that she is sure about what she knows, ``I know (at some time $t'$) that I observed outcome $a=1$ at time $t=1$, and I know (at $t'$) that I did not observe $a=0$ at time $t=1$.''  Indeed, in modal logic  \textbf{S} can be written as
$$K_{i,t'} K_{i, t} (x_i=x')  \implies K_{i,t'} K_{i,t} \ \neg (x_i\neq x'), \quad \forall \ t'$$
whenever $x_i$ is the outcome of a quantum measurement performed by agent $i$ at time $t$. 
But if we consider an agent who applies unitary quantum theory to the whole lab, including herself, then her introspection depends on whether she wants to talk about her ``local'' perspective of the state in the lab after experiencing a given outcome, or the  ``global '' state of the lab as seen from the outside. Alice, for example, may measure $R$ and then say ``locally I observed outcome $a=1$, but I know that globally what happened is that my experience of the outcome became entangled with $R$, that is my memory and $R$ are in the joint state $\sqrt{\frac{1}{3}}\ket{0}_R\ket{0}_A+\sqrt{\frac{2}{3}}\ket{1}_R\ket{1}_A$.'' Thus, Alice  could first choose  a \emph{context} (local to her perspective or global), appropriate to the situation in which she is making a statement, and only then produce the statement. If Alice keeps the global perspective, she will be able to predict that Wigner may obtain the outcome $w=ok$ even when she subjectively experiences $a=1$, thus avoiding the paradox. In general, in order to keep consistency, propositions must be qualified with adequate contexts. However, we should note that keeping track of the context quickly becomes unpractical in settings with several agents and measurements, as it requires agents to keep a record in their memory that is exponentially large on the number of measurements. We formalize context in Section~\ref{sec:suggestions}.

\paragraph{Truth and trust.}
The  Frauchiger-Renner assumption \textbf{C} is reflected in the distribution and truth axioms of modal logic. 
The distribution axiom (agents can chain their reasoning)
generalizes \textbf C, which corresponds to the special case of consistency of the statements of the agents.
The truth axiom says that if a fact $\phi$ is always true for one agent, then it must be true for all agents. It allows agents to simplify their reasoning, e.g.\ Bob simplifies  ``I know that Alice knows $\phi$'' to ``I know $\phi$'', and
it implies  trust in the universality of agents experiences.
Thought experiments like Wigner's friend and  Frauchiger-Renner suggest that this is not the case for quantum agents, and that a fully quantum generalization of modal logic must drop the truth axiom.  This can be easily seen in the Wigner's friend scenario, where a statement about a definite outcome can be true for an observer inside the lab, but its truth value  is not well-defined for Wigner himself. We return to trust in Section~\ref{sec:suggestions}.

\subsection{Reformulating the Frauchiger-Renner paradox}

We can see how the combined assumptions of classical modal logic correspond roughly to the Frauchiger-Renner assumptions, and lead to the same contradiction.

\begin{theorem}[Kripke structures in unitary quantum theory are not sound]
\label{thm:nogoKripke}

It is possible to find a contradiction in Kripke structures equipped with the axioms of knowledge satisfying the following: 
\begin{enumerate}
    \item Agents are allowed to perform measurements on quantum systems. Each agent's memory be modelled by other agents as a  quantum system that may itself be measured. 
    \item The measurement statistics satisfy the Born rule, or at least the weak version of the Born rule as in assumption {\bf Q}. In particular, the truth assignments satisfy this.
    \item When an agent performs a measurement, this is modelled by the other agents as a reversible evolution on their lab (assumption {\bf U}).

    \item Possible worlds can be parametrized by outcomes of quantum measurements. In particular, when an agent performs a measurement, they may perform positive and negative introspection on the observed outcome (similarly to assumption {\bf S}). 
    \item Agents are allowed to know the setting of an experiment beforehand. 
\end{enumerate}

All of the above should be included in the common knowledge that determines the Kripke structure in order to find a contradiction. 

Specifically, we find that there exists a possible world in the Kripke structure where its axioms are not sound (Appendix~\ref{appendix:modal}).

\end{theorem}

\begin{proof}
We take $M$ to be the Kripke structure describing the setting of the Frauchiger-Renner experiment and all agents involved:  it includes \textbf Q, \textbf U and the protocol of the experiment as part of the common knowledge. The proof follows the exposition of the contradiction in the first part of this paper, except that now we apply the axioms of knowledge. 
\begin{enumerate}
    \item Using \textbf S, we can say that in all worlds where Alice sees outcome $a=1$ she can predict $w=fail$, since in this setting the agents always make the same measurements, and, due to assumptions \textbf Q and \textbf U, the conclusion is always valid for these measurements. That is,  $$(M,s)\models K_A (a=1 \implies w=fail), \quad \forall \ s.$$
    
    \item It follows from the truth axiom that 
    $$ (M,s)\models  (a=1 \implies w=fail), \quad  \forall \ s,$$
    which removes one layer of reasoning from the statement.
    
    \item Since the above holds for all worlds $s$, we can apply the knowledge generalization rule to obtain that for all agents $i$,  $\models K_i  (a=1 \implies w=fail), \forall \ i$. 
    
    \item Applying the same reasoning to the other statements, we obtain
    \begin{align*}
       &\models K_i  (u=ok \implies b=1), \\
      &\models K_i  (b=1 \implies a=1),\\
      &\models K_i  (a=1 \implies     w=fail), \qquad \forall i. 
    \end{align*}

    \item Using the distribution axiom, we obtain the common knowledge $$\models K_i  (u=ok \implies w=fail), \forall i.$$
    
    \item Plugging in the fact that the experiment halts in a particular world $s$ where Wigner knows that the final outcomes are $w=ok, u=ok$, we obtain $$(s,M)\models K_W [u=ok \wedge w=ok \wedge  (u=ok \implies w=fail)].$$
    
    \item Now we use the distribution axiom to simplify this to 
\begin{align*}
     (s,M)&\models K_W [u=ok \wedge w=ok \wedge  (u=ok \implies w=fail)] \\
      \Leftrightarrow \ 
     (s,M)&\models K_W [u=ok\wedge  (u=ok \implies w=fail)   \wedge w=ok ] \\
      \Leftrightarrow \ 
     (s,M)&\models K_W [ w=fail  \wedge w=ok    ].
\end{align*}

    \item Using again the truth axiom, we  obtain   $$(s,M)\models  [w=ok \wedge w=fail].$$ 
    \item Using ${\bf S }$, $``(w=fail) \implies \operatorname{not} (w = ok),"$  we obtain a logical contradiction.
    
\end{enumerate}

\end{proof}

\subsection{Suggestions to generalize modal logic}
\label{sec:suggestions}

We may now attempt to implement the suggestions introduced in Section~\ref{sec:quantumproblems} to relax the axioms of classical epistemic modal logic, and test them in the Frauchiger-Renner setting. 
As a first step, we will replace the truth axiom with a subjective notion of trust between agents. We will see that, while necessary, this does not suffice --- we still obtain a contradiction. We also explore possible ways to implement the idea of contexts.

\subsubsection{Replacing truth with trust.}
\label{subsec:trust}
The truth axiom says that if an agent ``knows'' something, it is objectively true. As we saw, in quantum theory this does not always hold.

\begin{definition}[Trust]
    We say that an agent $i$ trusts an agent $j$ (and denote it by $j \leadsto i$ ) if and only if $$(M,s)\models K_i \ K_j \ \phi \implies K_i \ \phi,$$ for all $\phi, s$. 
\end{definition}

\paragraph{Trust is not transitive.} We may equip a Kripke structure with trust relations. These are not necessarily symmetric, or transitive.\footnote{In particular, the trust structure is not necessarily a pre-order in the set of agents.}
To see this, consider the case $A \leadsto B \leadsto C$, that is Alice trusts Bob who trusts Charlie. If  $K_A K_C \phi$ (Alice knows that Charlie knows $\phi$), it does not necessarily follow that $K_A \phi$, as Alice could think that Charlie is wrong or that their experience is not transferable, as in the case of observing outcomes of quantum measurements; for this we would need the intermediate agent Bob, as $ K_A K_B K_C \phi \implies K_A K_B \phi  \implies K_A \phi$. However, if there was no way for Bob to be sure that Charlie indeed knows $\phi$, then the trust chain would not reach Alice. If this situation seems far-fetched, think of the following closer-to-life instantiation of the example: Alice is a newspaper editor, Bob one of her journalists, and Charlie one of Bob's secret sources inside a criminal organization. The key here is that Alice trusts everything that Bob knows, independently of how he reached that knowledge (and in particular if he learned $\phi$ through Charlie, it suffices for Alice to know that Bob knows $\phi$). On the other hand, if Charlie went directly to Alice with some information, Alice, who doesn't know that Charlie is Bob's trusted source, would have no reason to trust it; indeed if we apply the trust condition in a different order, $ K_A  K_B ( K_C \phi ) \implies K_A K_C \phi$, we cannot reach the final conclusion $K_A \phi$.

\paragraph{Trust relations in the Frauchiger-Renner experiment.}
In the Frauchiger-Renner setting, we can specify for each agent who they trust, and at what point in the experiment. We have the following trust structure:
\begin{align}
    A_{t=0,1,2} \leadsto B_{t=2} \leadsto U_{t=3}\leadsto  W_{t=4} \leadsto A_{t=0,1,2} .
    \label{eq:trustFR}
\end{align}
Note that for example, Ursula at time $t=3$ may think of Alice as simply a quantum system to be measured and not an agent full capable of reasoning. However, Alice's reasoning about the outcomes of the experiment will reach Ursula via Bob. 
We can have additional trust relations, for example $A_{t=1} \leadsto W_{t=1,2}$. However, they are not necessary to derive the contradiction.

\paragraph{Trust is necessary for descriptions of quantum measurements.} In the Frauchiger-Renner setup, Ursula trusts Bob but not Alice. In particular, Ursula may model  Alice's memory as a quantum system $A$, and Alice's measurement of a quantum system $R$, as unitary evolution on the joint system $R\otimes A$, without referring to Alice's experience of observing an outcome $a$. We can formalize the statement ``Alice obtained the outcome $a=0$ at time $t=1$,  as $K_{A_{t=1}} (a=0)$, where $A_{t=1}$ represents agent Alice at that time: that is, at time $t=1$, Alice knows that she observed outcome $a=0$. We do not require that Ursula make such statements, for example, as she may want to avoid any proposition of the sort ``this quantum system $A$ `knows' something." Instead, she will get information that she can relate to via trusted agents, as we will see. Indeed, we will model the whole experiment such that agents only need to make statements of the form $K_i \phi$ about trusted agents --- Ursula will never have to model Alice's measurement outcome directly. 
Bob's statement ``If I observe $b=1$ I know that Alice must have $a=1$" can be  formalized as 
$$\models K_{B_{t=0,1,2}} [K_{B_{t=2}} ( b=1) \implies K_{A_{t=1,2}}(a=1)],$$ 
where the time stamps refer to all the times at which agents may have this knowledge. 

\paragraph{Analysis of the Frauchiger-Renner experiment.} 
Theorem~\ref{thm:nogoTrust} below tells us that while replacing truth with trust is necessary to treat quantum settings, it is not sufficient to give us a sound logic system --- at least not with the trust structure (\ref{eq:trustFR}) and with assumption  {\bf S}. This was to be expected, as the original Frauchiger-Renner consistency assumption {\bf C} is  closer to the trust condition than to the truth axiom. One way out is to say that  (\ref{eq:trustFR})  is not a valid trust structure, which should be carefully justified: we'd need criteria that allow us to reject (\ref{eq:trustFR}) and still keep natural trust structures of more classical situations, all the while considering that all agents are physical, and ultimately quantum, systems. Another possible solution is to reject {\bf S}; here the challenge is to do it in a way that still lets us make non-trivial statements based on subjective observations. Formalizing context, in the next section, is a first step in that direction.

\begin{restatable}[Subjective trust in the Frauchiger-Renner setting]{theorem}{ThmTrustNoGo}

\label{thm:nogoTrust}
In the setting of Theorem~\ref{thm:nogoKripke}, if we replace the truth axiom with the trust structure (\ref{eq:trustFR}) for the Frauchiger-Renner thought experiment, we obtain a contradiction. 
\end{restatable}
The proof can be found in Appendix~\ref{appendix:proofs}.

\paragraph{Trust and persistence of agents.} We can use trust relations to characterize the persistence of agent's memories over time --- which is particularly relevant in quantum settings. In our analysis above, for example, we had the trust relation $A_{t=1} \leadsto A_{t=2}$,  but not necessarily $A_{t=2} \leadsto A_{t=3}$. This is because Alice's memory is tampered with at time $t=3$, so agent ``Alice at time $t=3$'' does not necessarily trust agent ``Alice at time $t=2$.''  

\subsection{Context}

A classic way to incorporate time stamps and other context in the logic is connected to the notion of \emph{two-dimensional semantics} (see~\cite{sep-two-dimensional-semantics} for a review). It is a variant of possible world semantics that uses two (or more) kinds of parameters in the truth evaluation, rather than possible worlds alone.
The motivation for the introduction of additional parameter in order to determine a truth value is the fact that every proposition is true or false given a specific, often implicit context \cite{Kripke2007}.
An example would be a proposition $\phi=$ ``Hobbits live in the Shire''; if by ``the Shire'' the speaker meant an area, say, in England, then the statement would be false. However, if we were talking about the Shire on Middle-Earth, in the third age, the statement may be valid.
The notion of context can include time, place, topic of when or how the statement was produced --- we can see it as a fine-graining of the proposition, which for modularity purposes is registered in another space. If in a given proposition one of the parameters is not specified, then we assume that the proposition refers to all possible values of this parameter.

For simplicity, we will assume that in our case the parameters of the context are {\bf (1)} the agent who produces the proposition, and {\bf (2)} whether she produces it in the local or global context (recall the discussion from Section~\ref{par:intr&cont}).  

\begin{definition}[Context in Kripke structures]\label{def:context}~ 
The \textbf{context space} $\mathcal C$ is formed by elements of the form $$c= \begin{pmatrix} A_1,e_1\\ A_2,e_2 \\ \vdots\end{pmatrix},$$ where $A_1,A_2,...$ are  agents, and $e_1, e_2,...\in\{\ell, \text{g}\}$ specify whether the proposition $\phi$ made by a corresponding agent was meant locally or globally.\footnote{This choice of context space is specific to our setting.}
Propositions $\phi \in \Phi$ can themselves be  associated with contexts, in bundles of the form $ (c , \phi)  \in \mathcal C \times \Phi$, about which agents can make statements. 

A \textbf{truth interpretation} is now a function 
$\pi:\Sigma\times\Phi\times \mathcal C\to\{\textbf{true},\textbf{false}\}$, which defines a truth assignment for a proposition $\phi\in\Phi$ in a possible world $s\in\Sigma$, given a context $c\in C$.

To model \textbf{agents' knowledge}, we must also take the context into account explicitly. First, we note that  the binary equivalence relation $\mathcal K_A^e$ is now parametrized by the context $c=(A, e)$.
Therefore the statement `an agent $A$ ``knows'' $\phi$ in a world $s$' is generalized to
$$(M,s) \models K_A^e\,(c, \phi), \ where \ c=\begin{pmatrix} A_1,e_1 \\ \vdots\\ A_n,e_n\end{pmatrix} \iff  \left[ 
 (s,t) \in \mathcal K_A^e\implies \pi (t, \phi,\begin{pmatrix} A_1,e_1 \\ \vdots\\ A_n,e_n\\ A,e \end{pmatrix} ) = {\bf true}
\right].
$$
This way, the relation between the knowledge operators and context is simple: every time we would (in traditional modal logic, using the truth axiom) remove one of the operators $K_i^{e_i}$, here we cautiously append it to the context: 
$$(M,s)\models K_C^{e_C}\ K_B^{e_B}\ \bigg(\begin{pmatrix} A_1,e_n\\ \vdots\\ A_n,e_2\end{pmatrix}, \ \phi\bigg) \implies (M,s)\models K_C^{e_C}\ 
\bigg(\begin{pmatrix} A_1,e_1\\ \vdots\\ A_n,e_n\\ B,e_B\end{pmatrix}, \ \phi\bigg),$$
for all choices of $\phi, s$ and $e_B$. 
Thus, every time the number of $K_i^{e_i}$ operators in a proposition is reduced, the dimension of the context vector is increased.
This operation does not reflect any trust relations between different agents; on the opposite, it simply relays the fact that a proposition carries along a context of how and by whom it was produced.   

In order to apply the \textbf{distribution axiom}, the appropriate contexts must match:   we reformulate the axiom as
$$(M,s) \models (K_A^e (c,\phi) \wedge  K_A^e  [ (c,\phi) \implies (c', \psi)]) \implies (M,s) \models K_A^e (c',\psi) .$$

\end{definition}

For example, consider the three statements: 
\begin{align*}
    1.\quad & (\begin{pmatrix} A_{t_2},\text{$\ell$}\end{pmatrix} , \  a=1), \\
    2.\quad & (\begin{pmatrix} A_{t_2},\text{g}\end{pmatrix}, \ \ket{\psi}_{RA}=\frac{1}{\sqrt{3}}\ket{00}_{RA}+\sqrt{\frac{2}{3}}\ket{11}_{RA}),\\
    3.\quad & ( \begin{pmatrix} A_{t_2},\text{$\ell$}\\ B_{t_3},\text{$\ell$}\end{pmatrix}, \ a=1) .
\end{align*}
The first corresponds to the statement ``$a=1$'', produced locally by Alice at time $t=2$.
The second is her statement at the same time $t=2$ ``My lab is in a global state $\frac{1}{\sqrt{3}}\ket{00}_{RA}+\sqrt{\frac{2}{3}}\ket{11}_{RA}$,'' which is meant in a global context, independently of her local observation of $a=1$. 
Finally, the third example is the chained statement made by Bob at time $t=3$ that  ``A at time $t=2$ could say `$a=1$','' locally:

\begin{align*}
    (M,s) \models K_{B_{t_3}}^{\ell} K_{A_{t_2}}^{\ell} (a=1) 
    &\implies \quad \models K_{B_{t_3}}^{\ell} (\begin{pmatrix} A_{t_2},\ell\end{pmatrix}, \ a=1) \\ 
    & \implies  \quad \models (\begin{pmatrix} A_{t_2},\ell\\ B_{t_3},\ell\end{pmatrix} , \ a=1).
\end{align*}

Due to the new restriction on the distribution axiom, the proofs of Theorems~\ref{thm:nogoKripke} and \ref{thm:nogoTrust} do not go through, and  as a result, the agents do not reach a contradiction.
We give an intuition of why we cannot reach the contradiction in Appendix~\ref{appendix:proofs}.

\begin{restatable}[Context in the Frauchiger-Renner setting]{claim}{ThmContext}
When taking \emph{context} into account, we do not reach a logical contradiction in the Frauchiger-Renner setting. 
\end{restatable}

The main problem with introducing context is the rigidity of the constraints to combine  statements; it seems to perplex even the simplest situations. For  example, take a classical setting where Alice asks Bob to tell her the color of a hat she's wearing. Let's suppose Bob tells her that it is red; in that case, the statement available to Alice is not ``My hat is red'', but rather ``My hat is red, by Bob's account''. For all Alice knows, Bob might be slightly colorblind, so she is always keeping in mind that this point of view is his and his only. Subsequently, that does not allow her to make a definitive conclusion about the color of her hat. 

There are ways for the constraint to be weakened, for example, by introducing \emph{trust} in a similar fashion as we did above, and allowing it to reduce the dimension of context vectors. This approach is in practice identical to what we've seen in section~\ref{subsec:trust}, and,  in the Frauchiger-Renner setting, with the trust relation introduced there, it leads to a paradox. In summary: at the moment there seems to be no middle ground between obtaining a paradox and not being able to perform even simple modal operations in classical settings. 
One possible way out is to say ``well, in the classical example above, the trust relation Bob$\leadsto$Alice is natural and valid, while in the Frauchiger-Renner setting there is a fundamental reason why that chain of trust is invalid.'' This frames the main problem as finding conditions for validity of trust relations.
Incidentally, another price to pay for logical consistency is \emph{complexity}: agents need to keep track of the context, which increases in size at every application of modal operators.

\section{Conclusions}
\label{sec:conclusions}

The take-home message of this paper is that we do not have a sound logic system to analyze agents' reasoning when quantum measurements are involved.
The most successful classical logical system, modal logic, fails to generalize at the level of each individual axiom, and candidate workarounds, like keeping track of the context of each statement, are unpractical, requiring exponentially large memories. A next step in this research is to work out conditions for when and how contexts can be combined, in order to recover the usual modal logic in classical settings. In the following, we review the contributions of this work  in more detail.

\paragraph{Completing the paradox.}
The first contribution of this work was to amend the Frauchiger-Renner result \cite{Frauchiger2018}. We identified one missing assumption, used in the proof of but not stated in the theorem. This assumption has to do with the type of evolution happening in an agent's lab, as seen by other agents.  We showed that without it the original proof does not hold, by taking the example of a collapse theory that satisfies the original assumptions \textbf Q,  \textbf S and  \textbf C but does not lead to a contradiction. 
Making the missing assumption explicit fixes the theorem, and helps classify how different interpretations of quantum theory fit into the paradox.

\paragraph{Analysis of classical logic.}
In the Frauchiger-Renner setup, agents construct statements by implementing  a certain logic system. These logical rules and axioms can be to some extent correlated with known classical axioms of knowledge; however, in the quantum case these rules  must  be weakened or extended.
In the second part of this paper, we reviewed classical modal logic, which is successfully used in classical multi-agent settings to reason about knowledge. 
We compared its axioms to the Frauchiger-Renner assumptions and saw why the classical axioms do not generalize to the fully quantum case. We formalized this by applying modal logic to the  Frauchiger-Renner and rederiving a paradox, which tells us that the resulting structure is not sound. That is, the semantics of modal logic, when combined with axioms of quantum theory \textbf Q and \textbf U, leads to a contradiction. While modal logic is a powerful tool to process multi-agent settings, it is also a subject to strict requirements on the rationality of the agents. As we have seen, in the quantum case the agents cannot even in principle be considered reliable. We also do not have sufficient criteria for the type of setups in which it is suitable to use classical logic.

\paragraph{Proposal to generalize modal logic.}
We showed that even if we drop the truth axiom and replace it with weaker, subjective trust relations among agents, we still obtain a contradiction in the Frauchiger-Renner experiment (with the trust relation of Eq.~\ref{eq:trustFR}). One could argue that it is that particular trust relation that is not  valid, and the task then becomes to identify criteria for valid trust relations: intuitive criteria that would include trust in traditional  classical settings but rule out relations that lead to contradictions in quantum setups.
One way to formalize this  is to keep track of the context in which statements are taken, apply all logical steps possible, and only then then use valid trust relations to combine and compress statements and contexts. Without trust relations, context-tracking gives us logical consistence at a high price in complexity: 
we lose the compressing power of modal logic, and the ability to combine most statements; in addition,  agents need memories that grow exponentially on the number of statements considered.

\paragraph{Relation to other work.} 
The general search for a structure which would adequately embed the propositions made in a quantum mechanical context, was started by von Neumann \cite{vonNeumann1955, vonNeumann1936}. He aimed it to `resemble the usual calculus of propositions with respect to \emph{and}, \emph{or}, and \emph{not}' \cite{vonNeumann1936}.
In motivating the need for quantum logic systems, Putman~\cite{Putnam1969} highlights that the emerging framework should recover the strength of  classical logic in the adequate limit (we should still be able to use it to solve puzzles such as Example~\ref{fig:hats}), and that in addition it should explain the counter-intuitive aspects of quantum theory: settings that result in paradoxes when analyzed through classical logic (such as the Frauchiger-Renner experiment) should not lead to contradictions through the lens of quantum logic \cite{Putnam1969}.  Current proposals for quantum logic systems \cite{QLStanford, Chiara2013part1, Chiara2013part2,vonNeumann1955, vonNeumann1936, DELStanford,Baltag2011, Bacciagaluppi2009} fall short of these two goals. 
Indeed, according to the review of the  quantum logic field, carried out by Bacciagaluppi in 2009 \cite{Bacciagaluppi2009}, 
current proposals at quantum logic can replace classical logic for particular experiments involving quantum systems; however, it is still unclear whether they could replace classical logic on a global scale, and most importantly, they are not powerful enough to solve quantum mechanical paradoxes. 
This is because  most attempts at quantum logic to date  are primarily focused on stripping down the Boolean structure of classical logic, and as far as we are aware they do not concern settings with multiple agents' perspectives, in the way that modal logic does.

For example, in
\emph{A quantum computational semantics for epistemic logical operators} by Beltrametti et al \cite{Chiara2013part1, Chiara2013part2}, truth values are not represented as bits but as qubits, and truth perspectives of individual agents are identified with a change of basis, while logical connectives are associated with quantum gates. However, such interpretation of truth values of propositions deprives them of their deterministic and binary nature, which we would like to keep when arguing about agents' reasoning.

Alternatively, one can use dynamic epistemic logic \cite{DELStanford}, which considers the consequences of public and private announcements that define the dynamics of the Kripke structure. Some work has been made in this direction \cite{Baltag2011}; however, it does not encompass situations where agents think about other agents' outcomes, and thus the application of the model is not quite clear.

One can adapt another types of logic to this situation, for example, paraconsistent logic \cite{PRStanford}, which draws an inspiration in a dialetheistic view \cite{DialetheismStanford}. This view supports the idea of the existence of true contradictions; thus, paraconsistent logic denies the principle of explosion which allows proof of any proposition via its negation. Nesting the contradictions right in the theory might give rise to substantial development of how logic would work for paradoxes similar to the one in Frauchiger-Renner's setup. Research on applications of paraconsistent logic to quantum settings is still embryonic.

\paragraph{Pedagogical contributions.}
As a bonus contribution to the community, this article presented a more pedagogical introduction to the Frauchiger-Renner experiment and to classical modal logic than we could find elsewhere in the literature. Finally, to settle one of the most common questions that the paper has raised in the last two years: the {\bf ch} in Frau{\bf ch}iger is pronounced as in Lo{\bf ch} Ness.

\section*{Acknowledgements}
We thank Renato Renner and Sandu Popescu for  discussions, and the QPL and TQC referees for feedback on the first version of this manuscript. 
LdR acknowledges support from the Swiss
National Science Foundation through 
SNSF project No.\ $200020\_165843$ and through the 
the National Centre of
Competence in Research \emph{Quantum Science and Technology}
(QSIT), and from  the FQXi grant \emph{Physics of the observer}.


\newpage 
\appendix

\addcontentsline{toc}{section}{\sc{Appendix}}

\section*{\textsc{Appendix}}

\section{Modal logic in more detail}
\label{appendix:modal}

\subsection{Application of modal logic to a classical setting}

Here we provide a solution to the puzzle in Example~\ref{fig:hats}.
Let us follow the inference of Arren. First, he considers what Ged thinks of this whole setting. Ged sees both Tehanu and Arren, and there are four possibilities: Tehanu wears black and Arren wears white, Tehanu wears white and Arren wears black, they both wear white or, finally, they both wear black. However, in the latter case, Ged would have all the reasons to announce that his hat is white; as he didn't do that, both Arren and Tehanu can conclude that it is not the case. Second, Arren considers Tehanu's thinking; if Arren wears black then Tehanu can only wear a white hat, and thus announce that fact. However, he does not, and so Arren concludes that he wears a white hat.

One could attempt to formalize Arren's thinking. Let us denote the state of the system as a set $(*,*,*)$ where the first element corresponds to Ged's state, second - to Tehanu's state, and third - to Arren's state. A black hat will be marked as 1, and a white hat as 0.

Initially, Arren has no information about his or Ged or Tehanu's state, so it is reasonable for him to assume all possible combinations - there are 7 in total, namely, $(0,0,0)$,$(0,0,1)$,$(0,1,0)$,$(1,0,0)$,$(0,1,1)$,\\$(1,0,1)$,$(1,1,0)$. His mind can be represented as a table which consists of three columns: first one portraying his assumption of what the situation is, and two last ones corresponding to what would Tehanu and Ged think in that case.  

In the beginning of the puzzle Arren's table looks like this:
\medskip\\
\begin{tabular}{ |p{3cm}|p{7cm}|p{4cm}|  }
\hline
\multicolumn{3}{|c|}{The initial state ($t=0$), according to Arren} \\
\hline
State &What Tehanu thinks &What Ged thinks\\
\hline
$(0,0,0)$		&$(0,0,0)$,$(1,0,0)$,$(0,1,0)$,$(1,1,0)$	&$(0,0,0)$,$(1,0,0)$\\
\hline
$(0,0,1)$		&$(0,0,1)$,$(1,0,1)$,$(0,1,1)$			&$(0,0,1)$,$(1,0,1)$\\
\hline
$(0,1,0)$		&$(0,0,0)$,$(1,0,0)$,$(0,1,0)$,$(1,1,0)$	&$(0,1,0)$,$(1,1,0)$\\
\hline
$(1,0,0)$		&$(0,0,0)$,$(1,0,0)$,$(0,1,0)$,$(1,1,0)$	&$(0,0,0)$,$(1,0,0)$\\
\hline 
$(0,1,1)$		&$(0,0,1)$,$(1,0,1)$,$(0,1,1)$			&$(0,1,1)$\\
\hline
$(1,0,1)$		&$(0,0,1)$,$(1,0,1)$,$(0,1,1)$			&$(0,0,1)$,$(1,0,1)$\\
\hline
$(1,1,0)$		&$(0,0,0)$,$(1,0,0)$,$(0,1,0)$,$(1,1,0)$	&$(0,1,0)$,$(1,1,0)$\\
\hline
\end{tabular}
\medskip\\
It can be clearly seen here that Ged can make an announcement only if the state is $(0,1,1)$; in all other cases he can wear both white and black hat from his point of view (last column). After he does not say anything, we have to exclude $(0,1,1)$ from everywhere in the table, given that Arren supposes Tehanu to use the same reasoning. 
That is, we have $K_A K_T K_G\ \neg (0,1,1)$, and by the truth and distribution axioms, $K_A \ \neg (0,1,1)$.
The table then transforms into:
\medskip\\
\begin{tabular}{ |p{3cm}|p{7cm}|p{4cm}|  }
\hline
\multicolumn{3}{|c|}{The state as Ged remains silent ($t=1$), according to Arren} \\
\hline
State &What Tehanu thinks &What Ged thinks\\
\hline
$(0,0,0)$		&$(0,0,0)$,$(1,0,0)$,$(0,1,0)$,$(1,1,0)$	&$(0,0,0)$,$(1,0,0)$\\
\hline
$(0,0,1)$		&$(0,0,1)$,$(1,0,1)$					&$(0,0,1)$,$(1,0,1)$\\
\hline
$(0,1,0)$		&$(0,0,0)$,$(1,0,0)$,$(0,1,0)$,$(1,1,0)$	&$(0,1,0)$,$(1,1,0)$\\
\hline
$(1,0,0)$		&$(0,0,0)$,$(1,0,0)$,$(0,1,0)$,$(1,1,0)$	&$(0,0,0)$,$(1,0,0)$\\
\hline 
$(1,0,1)$		&$(0,0,1)$,$(1,0,1)$					&$(0,0,1)$,$(1,0,1)$\\
\hline
$(1,1,0)$		&$(0,0,0)$,$(1,0,0)$,$(0,1,0)$,$(1,1,0)$	&$(0,1,0)$,$(1,1,0)$\\
\hline
\end{tabular}
\medskip\\
Now Arren can turn his attention to Tehanu's column; if the state of the system is $(0,0,1)$ or $(1,0,1)$ then Tehanu can announce that she is wearing a white hat; however, she does not, thus we can eliminate these possibilities. That is, we have $K_A K_T\  \neg (x,0,1), \ \forall \ x$.
\medskip\\
\begin{tabular}{ |p{3cm}|p{7cm}|p{4cm}|  }
\hline
\multicolumn{3}{|c|}{The state as Ged and Tehanu remain silent ($t=2$), according to Arren} \\
\hline
State &What Tehanu thinks &What Ged thinks\\
\hline
$(0,0,0)$		&$(0,0,0)$,$(1,0,0)$,$(0,1,0)$,$(1,1,0)$	&$(0,0,0)$,$(1,0,0)$\\
\hline
$(0,1,0)$		&$(0,0,0)$,$(1,0,0)$,$(0,1,0)$,$(1,1,0)$	&$(0,1,0)$,$(1,1,0)$\\
\hline
$(1,0,0)$		&$(0,0,0)$,$(1,0,0)$,$(0,1,0)$,$(1,1,0)$	&$(0,0,0)$,$(1,0,0)$\\
\hline 
$(1,1,0)$		&$(0,0,0)$,$(1,0,0)$,$(0,1,0)$,$(1,1,0)$	&$(0,1,0)$,$(1,1,0)$\\
\hline
\end{tabular}
\medskip\\
In all remaining scenarios, Arren wears a white hat, so he can safely announce it at $t=3$.

Thus, here a thinking agent considers all possible worlds that can happen, and step by step, he eliminates them from the picture, and finally narrow their variety down, being able to make a definite conclusion.

\subsection{Principles, development and other logic systems}

\paragraph{Soundness and completeness.}
The aim of any logic system is to distinguish between valid and invalid statements; the whole set of axioms and rules is designed in order to prove whether a given statement is true or false. The main challenges faced by the logician developing such a system are to make sure that the said system is \textit{sound} and \textit{complete} \cite{LogicStanford}. When we say that the system is \textit{sound}, we mean that any statement that is wrong cannot be proven right, or, in other words, there are no false positives; \textit{completeness} means that all statements that are claimed to be valid can be in principle proven in the system.
We have to emphasize here that not all statements in the set of propositions can be validated; one can connect it to the G{\"o}del's first incompleteness theorem, which states that any consistent formal system is incomplete in a sense that there exist statements which cannot be proven right or false \cite{Godel1931}.

\paragraph{Development and related logic systems.}
If we want to describe an agent who is able to make predictions about the outcomes of an experiment, we have to, first of all, have the means to model the knowledge the agent possesses. It is one of the central concepts that can be used to describe the experiments similar to the Wigner's friend experiment and its extended version, where the agents are actively using the information they have to draw certain conclusions. In other words, we can say that they store statements and combine them according to the rules established before the start of the experiment.
In general, the observers performing the experiment operate and communicate using the statements of the type ``it is certain that'' or ``it is possible that'', in other words, these expressions qualify the truth of a statement. This idea resembles the concept of possible worlds, introduced by Leibniz \cite{Leibniz1771}, and followed by the development of modal logic \cite{Lewis1918, Lewis1959, Carnap1988, Goldblatt2006, Kripke2007}. The language of the modal logic allows for the formal analysis of the occasionally complex linguistic constructions, and epistemic logic, which is a special type of modal logic, deals with agents discussing possible worlds of the structure.
Alternatively, as the agents express their subjective beliefs and predictions in their arguments, the subjective logic is introduced by \cite{Josang2001}, which assigns to propositions a level of belief, of disbelief and of uncertainty, and establishes update rules, as well as the interaction of agents and trust networks. 
Another approach, called Dempster-Shafer theory, is based on two ideas: obtaining degrees of belief for one question from subjective probabilities for a related question, and Dempster's rule for combining such degrees of belief when they are based on independent items of evidence. In essence, the degree of belief in a proposition depends primarily upon the number of answers (to the related questions) containing the proposition, and the subjective probability of each answer. The theory assigns a degree of belief and of plausibility for each statement, which serve as two bounds for belief in the statement, and describes update rules for collecting and forgetting information \cite{Yager2008, Fritz2015}. 
Finally, fuzzy logic, a form of many-valued logic, also allows for appointing fuzzy truth values between 0 and 1 \cite{Novak2012, Zadeh1996, Zimmermann1996}. It is connected to the concept of partial truth, where the truth value can range between being absolutely true or absolutely untrue.

\section{Proofs}
\label{appendix:proofs}
In this appendix you may find the proofs of some of the results in the main text.

\ThmTrustNoGo*

\begin{proof}
Even before the experiment starts, we can make the following statements, reflecting what counterfactual reasoning of the different agents about future measurements: 
\begin{align*}
  \models   K_{A_{t<3}} &[K_{A_{t=1}} (a=1) \implies K_{W_{t=4}} (w=fail)], \\
   \models  K_{B_{t<4}} &[K_{B_{t=2}} (b=1) \implies K_{A_{t=1,2}} (a=1)], \\
    \models K_{U_{t<5}} &[K_{U_{t=3}} (u=ok) \implies K_{B_{t=2,3}} (b=1)].
\end{align*}
If we add in the compatible assumption that agents can talk to each other before the experiment begins, we obtain in particular
\begin{align*}
   \models K_{B_{t=2}} K_{A_{t=1}} &[K_{A_{t=1}} (a=1) \implies K_{W_{t=4}} (w=fail)], \\
     \models K_{U_{t=3}}  K_{B_{t=2}} &[K_{B_{t=2}} (b=1) \implies K_{A_{t=1,2}} (a=1)], \\
     \models K_{W_{t=4}}  K_{U_{t=3}} &[K_{U_{t=3}} (u=ok) \implies K_{B_{t=2,3}} (b=1)], 
\end{align*}
and so on, up to third- and fourth-order statements like
\begin{align*}
   \models  K_{U_{t=3}}  K_{B_{t=2}} K_{A_{t=1}} &[K_{A_{t=1}} (a=1) \implies K_{W_{t=4}} (w=fail)], \\
\models  K_{W_{t=4}} K_{U_{t=3}}  K_{B_{t=2}} K_{A_{t=1}} &[K_{A_{t=1}} (a=1) \implies K_{W_{t=4}} (w=fail)], \\
     \models K_{W_{t=4}}  K_{U_{t=3}}  K_{B_{t=2}} &[K_{B_{t=2}} (b=1) \implies K_{A_{t=1,2}} (a=1)]. 
\end{align*} 
There are many ways to apply the trust relations and distribution axiom in order to combine the above statements. For now, we would prefer for each agent's direct statements to be about the knowledge of trusted agents: for example, we want to avoid statements of the form $K_U K_A \phi$, but rather go the route  $K_U K_B K_A \phi \to K_U  K_B  \phi \to K_U \phi$. 
To do so, let us take all the statements that start with Wigner's knowledge, 
\begin{align*}
\models K_{W_{t=4}}  K_{U_{t=3}} &[K_{U_{t=3}} (u=ok) \implies K_{B_{t=2}} (b=1)],\\
 \models K_{W_{t=4}}  K_{U_{t=3}}  K_{B_{t=2}} &[K_{B_{t=2}} (b=1) \implies K_{A_{t=1}} (a=1)],\\
\models  K_{W_{t=4}} K_{U_{t=3}}  K_{B_{t=2}} K_{A_{t=1}} &[K_{A_{t=1}} (a=1) \implies K_{W_{t=4}} (w=fail)].
\end{align*}
First we combine the  bottom two statements,
\begin{align*}
 \models K_{W_{t=4}}  K_{U_{t=3}}  K_{B_{t=2}} \Big(\quad&  \Big[ K_{B_{t=2}} (b=1) \implies K_{A_{t=1}} (a=1) \Big] \\ 
 \wedge &\Big[ K_{A_{t=1}} [K_{A_{t=1}} (a=1) \implies K_{W_{t=4}} (w=fail)]\Big]\Big), 
\end{align*}
And now we apply the trust condition $A_{t=1}\leadsto B_{t=2}$,
\begin{align*}
 \models K_{W_{t=4}}  K_{U_{t=3}}  K_{B_{t=2}} \Big(\quad & 
 \Big[ K_{B_{t=2}} (b=1) \implies K_{A_{t=1}} (a=1) \Big]\\
 \wedge & \Big[ K_{A_{t=1}} (a=1) \implies K_{W_{t=4}} (w=fail)\Big]\Big), 
\end{align*}
followed by the distribution axiom for Bob, 
\begin{align*}
 \models K_{W_{t=4}}  K_{U_{t=3}}  K_{B_{t=2}} &[K_{B_{t=2}} (b=1) \implies K_{W_{t=4}} (w=fail)]. 
\end{align*}
Now we combine this with the remaining statement about Wigner's knowledge, 
\begin{align*}
 \models K_{W_{t=4}}  K_{U_{t=3}}   \Big( \quad & \Big[K_{U_{t=3}} (u=ok) \implies K_{B_{t=2}} (b=1)\Big]\\  \wedge  & \Big[ K_{B_{t=2}} [K_{B_{t=2}} (b=1) \implies K_{W_{t=4}} (w=fail)]\Big]\Big). 
\end{align*}
Again, we apply the trust condition $B_{t=2}\leadsto U_{t=3}$, followed by the distribution axiom for Ursula, obtaining
\begin{align}
 \models K_{W_{t=4}}  K_{U_{t=3}}   &[K_{U_{t=3}} (u=ok) \implies  K_{W_{t=4}} (w=fail)].  
 \label{eq:chainInference}
\end{align}
Note that the above holds for all possible worlds, and that Wigner and Ursula can reach the conclusion $K_{U_{t=3}} (u=ok) \implies  K_{W_{t=4}} (w=fail)$ before the experiment even starts. 
Additionally, at the end of the experiment, Ursula and Wigner learn of each other's outcomes, in particular
$$ \forall \ s: \qquad (M,s)  \models K_{U_{t=3}} (u=ok)  \implies (M,s)  \models K_{W_{t=4}} K_{U_{t=3,4}} (u=ok).$$
Now we run the experiment and analyse the case $u=w= ok$, that is
$$\exists \ s: \qquad (M,s)  \models K_{U_{t=3,4} }(u=ok) \wedge K_{W_{t=4} }(w=ok).$$ Given the above, we also have 
$$\exists \ s: \qquad  (M,s)  \models  K_{W_{t=4} } [K_{U_{t=3,4} }(u=ok) \wedge  K_{W_{t=4} }(w=ok)]. $$
Now we can combine this with (\ref{eq:chainInference}), obtaining
\begin{align*}
    \exists \ s: \quad  (M,s)  \models  K_{W_{t=4} } \Big( \quad &  \Big[K_{U_{t=3} }(u=ok) \wedge  K_{W_{t=4} }(w=ok)\Big] \\  \wedge&  K_{U_{t=3}}   \Big[K_{U_{t=3}} (u=ok) \implies  K_{W_{t=4}} (w=fail)\Big]   \Big),
\end{align*}
and applying the trust relation $U_{t=3} \leadsto W_{t=4}$, and the distribution axiom again, we get 
\begin{align*}
    \exists \ s: \quad  (M,s)  \models  K_{W_{t=4} } [   K_{W_{t=4} }(w=ok)\wedge  K_{W_{t=4}} (w=fail)].
\end{align*}
Finally, we can apply the trivial trust relation $W_{t=4} \leadsto W_{t=4}$ to obtain 
\begin{align*}
    \exists \ s: \quad  (M,s)  \models  K_{W_{t=4} } [   w=ok\wedge   w=fail].
\end{align*}
To obtain the contradiction, we use {\bf S}, that is $$``(w=fail) \implies \operatorname{not} (w = ok)."$$ 
\end{proof}

\ThmContext*

We do not prove this claim directly; we only give an intuition of why we cannot reach the contradiction.

\paragraph{Proof sketch.}~

We start from the point of the proof of Theorem~\ref{thm:nogoTrust} where we combine two statements about Wigner's knowledge,
\begin{align*}
 \models K_{W_{t=4}}^\ell  K_{U_{t=3}}^\ell  K_{B_{t=2}}^\ell \Big(& \quad \Big[ K_{B_{t=2}}^\ell (b=1) \implies K_{A_{t=1}}^\ell (a=1) \Big]\\ & \wedge \Big[ K_{A_{t=1}}^\ell [K_{A_{t=1}}^\ell (a=1) \implies K_{W_{t=4}}^\ell (w=fail)]\Big]\Big).
\end{align*}

Then we compress this expression in a way that has been referred to in the Definition~\ref{def:context}, in order to obtain explicit structure of the context of statements concerning Wigner's knowledge,
\begin{align*}
\models \quad  
    & \Big(
        \begin{pmatrix} B_{t=2},\ell \\ U_{t=3},\ell \\ W_{t=4}, \ell \end{pmatrix}, \quad
        [K_{B_{t=2}}^\ell(b=1) \implies K_{A_{t=1}}^\ell (a=1)]
    \Big ) \\ 
\wedge \ 
    &\Big(
        \begin{pmatrix} A_{t=1},\ell \\ B_{t=2},\ell \\ U_{t=3},\ell \\ W_{t=4}, \ell \end{pmatrix}, \quad 
        [K_{A_{t=1}}^\ell(a=1) \implies K_{W_{t=4}}^\ell (w=fail)]
    \Big)
\end{align*}
At this point we can no longer apply the distribution axiom, due to the extra context element on the second term.
Thus, here the contradiction cannot be reached.

\section{Interpretations in more detail}
\label{appendix:interpret}
Here we give short introductions to each of the interpretations discussed in Section \ref{sec:interpretations}.

\paragraph{Copenhagen interpretation.}
Quantum theory keeps the core of classical physics as assigning certain mathematical structures to points in time and space; however, the behavior of said structure is different, as well as its meaning; it no longer represents a material, but rather an informational structure. This can be summarized in the famous words by Heisenberg \cite{Heisenberg1958}:
\begin{displayquote}
In consequence, we are finally led to believe that the laws of nature which we formulate mathematically in quantum theory deal no longer with the particles themselves but with our knowledge of the elementary particles. The question of whether these particles exist in space and time ``in themselves" can thus no longer be posed in this form. We can only talk about the processes that occur when, through the interaction of the particle with some other physical system such as a measuring instrument, the behavior of the particle is to be disclosed. The conception of the objective reality of the elementary particles has thus evaporated in a curious way, not into the fog of some new, obscure, or not yet understood reality concept, but into the transparent clarity of a mathematics that represents no longer the behavior of the elementary particles but rather our knowledge of this behavior.
\end{displayquote}

In order to deal with technical difficulties that arise within the framework where physical states of the material world are substituted by actions of obtaining knowledge, a new philosophy was created by Heisenberg, Pauli, Born and Bohr: the Copenhagen interpretation.

The essence of the Copenhagen interpretation, as described by Bohr \cite{Bohr1934, Bohr1958}, is:
\begin{displayquote}
In our description of nature the purpose is not to disclose the real essence of phenomena but only to track down as far as possible relations between the multifold aspects of our experience. [1934]
\end{displayquote}
\begin{displayquote}
[...] the appropriate physical interpretation of the symbolic quantum-mechanical formalism amounts only to predictions, of determinate or statistical character, pertaining to individual phenomena appearing under conditions defined by classical physical concepts. [1958]
\end{displayquote}
\noindent
In his work Bohr also points out that the essential part of scientists' goal is to communicate the details of the setup they have used in experiment and the contents of the obtained outcomes, which they are bound to express in concepts of classical physics \cite{Bohr1958}:
\begin{displayquote}
As the goal of science is to augment and order our experience, every analysis of the conditions of human knowledge must rest on considerations of the character and scope of our means of communication. [...] In this connection, it is imperative to realize that in every account of physical experience one must describe both experimental conditions and observations by the same means of communication as one used in classical physics. 
\end{displayquote}
Bohr does not imply in this article that there is an objective reality that corresponds to the laws of classical physics.

The Copenhagen solution to the inconsistency present between classical and quantum physics is as follows: for the practical purpose of reasoning about experiments, we split the modelling of nature  into two parts. The first part would be the observing system, which includes the minds and bodies of the humans acquiring knowledge via the experiment, and also the measuring devices; the measurement setup and the following outcome experiences are described in ordinary language of classical physics. The agent (observer) thus should be able to say, for example, ``I switched on the device and ten minutes later I saw the pointer turn to the left''.
The second part of nature is the observed system, which would be described in the language of quantum mechanics.

Such separation between two parts of nature is called the \emph{Heisenberg cut}. Above the cut one uses classical descriptions, while below the cut one uses a quantum mechanical description. This cut can be moved, for example, from below the measuring device to above it, thus including the device into the quantum part of the system. In principle, one can even move the cut above a given observer, for example in multi-agent settings, as discussed here.

The observing system is left with a choice of an appropriate question to ask --- the mathematical description of the theory does not specify the question. The Copenhagen interpretation approaches this  by postulating conscious observers in a physical universe (able to communicate and having an interest in acquisition of knowledge), who are essentially free to make the choice of the question \cite{Bohr1958}:
\begin{displayquote}
The freedom of experimentation, presupposed in classical physics, is of course retained and corresponds to the free choice of experimental arrangement for which the mathematical structure of the quantum mechanical formalism offers the appropriate latitude.
\end{displayquote}

To give an example of how the Heisenberg cut can be moved, consider John von Neumann's understanding of quantum mechanics. Von Neumann named the physical posing of the question as \textit{Processes of type 1}, while so-called \textit{Processes of type 2} correspond to the subsequent evolution of the quantum state \cite{vonNeumann1955}.
A process of type 1 (choosing to perform a $Z$ measurement on a spin, for example) may have a part that is not physically modelled (the decision that takes place in the abstract ``mind'' of the agent), and a part that is physically modelled (the brain processes associated with that decision, the pressing of a button corresponding to the $Z$ measurement). The latter then affects the course of processes of type 2 (the physical interaction of the measurement apparatus with the spin). This allows us to move 
 the Heisenberg cut above the observer's experience of the outcome, leaving only the abstract ``choice of measurement'' in the classical realm.

\paragraph{Bohmian theory.}
The pilot-wave approach of Bohm \cite{Bohm1952} assumes that both things exist at once: the state of the physical universe specified by quantum theory (superposition of all possible configurations) and the classical world (a certain configuration that is there even when unobserved) which determines the contents of our consciousness. The quantum state evolves continuously according to von Neumann's processes of type 2, and the real classical world is enveloped in the quantum world in a way which makes all predictions of quantum theory correct. It is a deterministic hidden variable theory which is also non-local: for example, in some  multi-particle settings the velocities of  each single particle depends on the positions of all other particles.

The theory is based on the following postulates: first, there is a configuration of the universe, described by coordinates, which is an element of the configuration space. Its dynamics is governed by the guiding equation. Entities like electrons act like actual particles, their velocities at any moment fully determined by the pilot wave, which in turn depends on the wave function. In this view, each electron is like a surfer: it occupies a particular place at every specific moment in time, yet its motion is dictated by the motion of a spread-out wave \cite{Quanta}. Second, the configuration is distributed according to the usual probability density (square modulus of the wave function, which is, in its turn, governed by Schr{\"o}dinger's equation) at all times. This state (called quantum equilibrium) agrees with the results of the standard theory in the usual regime of single observers.

In his later works \cite{Bohm2008} Bohm makes an effort to explain consciousness of the observer in terms of his theory.
The problem this approach faces is somehow similar to the problem of the many-worlds: if the universe is governed by von Neumann's Process of type 2 alone, then the physical states of devices and brains are not partitioned, thus failing to produce distinct experiences depending on outcomes of experiments. For a short summary of the theory, see \cite{Sudbery2016}.

\paragraph{Relative-state formalism.}
Stefan Wolf and Veronika Baumann in their article ``On Formalisms and Interpretations'' \cite{Baumann2018} draw a comparison between textbook, collapse quantum mechanics and what they call ``the relative-state formalism'' (adapted from \cite{Everett1957}). They stress that the universal and unitary quantum theory, postulated by the relative-state formalism, calls for a formulation of alternative version of the Born rule, and, subsequently, provides another type of formalism, different from what the standard quantum mechanical model has to offer. 

The relative-state formalism reproduces the same probabilities as standard quantum mechanics for measurements on the same system, but is inequivalent to it in case of observers observing another observers, or, in other words, in case of a Wigner friend's-like scenario. Here we will briefly outline two formalisms presented in the paper, and then apply them to the Wigner's friend experiment.

According to the standard quantum mechanics, the Born rule, which gives the probability of a measurement result for a quantum state $\ket{\phi}$ and an observable $A=\sum_a a\ket{a}\bra{a}$, looks like:
\begin{gather*}
p_\phi (a)=Tr(\ket{a}\inprod{a}{\phi}\bra{\phi})=|{\inprod{a}{\phi}}|^2.
\end{gather*}
Now consider two observers $O_1$ and $O_2$ consecutively measuring a quantum system in the initial state $\phi$; their measurements are given by families of projectors $\{\ket{a}\bra{a}\}$ and $\{\ket{b}\bra{b}\}$ respectively. Then the conditional probability of result $b$ given $a$ is:
\begin{gather*}
p_\phi (b|a)=\frac{p_\phi (a,b)}{p_\phi (a)}=|{\inprod{b}{a}}|^2, \text{where } p_\phi (a,b) \text{ is the joint probability of } a \text{ and } b.
\end{gather*}

According to the relative-state formalism, measurements can be represented as isometries, correlating the state of the observer with the state of the observed system. For an observer $O$ (with an orthogonal set $\{\ket{A_a}\}$ recording the result) measuring a state $\phi$ of a system S with respect to a projector family $\{\ket{a}\bra{a}_S\}$ the isometry is
\begin{gather*}
V_O: \mathcal{H}_S\rightarrow\mathcal{H}_S\otimes\mathcal{H}_O, \ket{a}_S\mapsto\ket{a}_S\otimes\ket{A_a}_O .
\end{gather*}
The authors then introduce two postulates for the case of the relative-state formalism:
\begin{enumerate}
\item The probability of observing $a$ is given by \\
$q_\phi (a)=Tr(\mathbbm{1}_S\otimes\ket{A_a}\bra{A_a}_O\cdot V_O\ket{\phi}\bra{\phi}V_O^\dagger)$.
\item The joint probability of states is given by the trace of the tensor product of the projectors onto those states acting on the overall state.
\end{enumerate}
If we again consider two observers $O_1$ and $O_2$ consecutively measuring a quantum system in the initial state $\phi$, then according to the postulates, the conditional probability for the measurement result $b$ given $a$ is:
\begin{gather*} 
q_\phi (b|a)=\frac{q_\phi (a,b)}{q_\phi (a)}=\frac{Tr(\ket{A_a}\bra{A_a}\otimes\ket{B_b}\bra{B_b}V_{O_2}V_{O_1}\ket{\phi}\bra{\phi}V_{O_2}^\dagger V_{O_1}^\dagger)}{\sum_b Tr(\ket{A_a}\bra{A_a}\otimes\ket{B_b}\bra{B_b}V_{O_2}V_{O_1}\ket{\phi}\bra{\phi}V_{O_2}^\dagger V_{O_1}^\dagger)}.
\end{gather*}
It is then shown that the conditional probabilities for two observers consecutively measuring the same quantum system are same for standard quantum mechanics and relative-state formalism. However, if we consider additionally a superobserver SO measuring the joint system with $\{\ket{b}\bra{b}_{S,O}\}$ these two cases are proven to be inequivalent. 

Three different scenarios can be distinguished in this situation: first, when both the observer and the superobserver use the standard quantum mechanics measurement-update rule (authors call it \textit{objective collapse} model); second, all agents can use the relative-state formalism (\textit{no-collapse} model); and, finally, third, when the friend uses the standard quantum mechanics and Wigner (to whom the lab and the friend evolve unitarily) applies the relative-state formalism (\textit{subjective collapse} model). In the last case, if the friend calculates the conditional probabilities with a collapse model, while Wigner uses the relative-state formalism, they will give contradicting answers to the same question regardless of the friend's measurement result (for a detailed discussion see \cite{Baumann2018}).

\paragraph{Many-worlds.}
This interpretation assumes that the quantum state of the universe exists (and evolves under the rules of Process of type 2); the fact that we obtain specific outcomes is considered merely a subjective illusion, as the collapse events of the Copenhagen interpretation are no longer occurring. The conscious experiences of an outcome are side products of this continuous physical evolution; the state of the the universe splits, forming so-called ``branches'' in the original formulation of Everett \cite{Everett1957}.

However, the many-worlds interpretation faces various problems, including the problem of appropriate basis choice and branching events; as summarized by Zurek in the comparison he draws between many-worlds and Copenhagen \cite{Zurek2006}:
\begin{displayquote}
The similarity between the difficulties faced by these two viewpoints becomes apparent, nevertheless, when we ask the obvious question, ``Why do I, the observer, perceive only one of the outcomes?'' Quantum theory, with its freedom to rotate bases in Hilbert space, does not even clearly define which states of the Universe correspond to the ``branches''. Yet, our perception of a reality with alternatives --- not a coherent superposition of alternatives --- demands an explanation of when, where, and how it is decided what the observer actually records. Considered in this context, the Many Worlds Interpretation in its original version does not really abolish the border but pushes it all the way to the boundary between the physical Universe and consciousness. 
\end{displayquote}

Thus, the question of why the experimenter perceives only one outcome is this picture still remains. Additionally, quantum mechanical description nowhere implies that the state of mind cannot be a superposition of states; thus it is not paradoxical \cite{Penrose1999}.

The attempts to interpret probabilities in the many-worlds scenario have been made, for example, by Lev Vaidman. He postulates that, while the wave function of a world specifies what we ``feel'' in the world, the absolute value of the coefficient behind it specifies the illusion of probability: all possible outcomes still take place, yet there is no difference between our experience and the experience of an agent with genuine probability \cite{Vaidman2001, Vaidman2016}. 

\paragraph{QBism.}
QBism is an interpretation of quantum mechanics which puts the agent performing the measurement in the spotlight of the theory. It was developed by Carlton Caves, Christopher Fuchs and R{\"u}diger Schack \cite{Caves2002, Fuchs2014, Fuchs2016}. QBism argues that a quantum state represents not the physical, but the epistemic state of the observer who assigns it according to her view on her future experiences. The probabilities of different outcomes the agent can get are then simply viewed as that agent's degrees of belief in each of a variety of these outcomes.
QBism does not consider the quantum state to possess the complete information about the system; on the contrary, QBist observers each model their own state of the system, based on the knowledge available, and these states can, in principle, be different. This makes QBism a single-user theory, and thus the way of performing the analysis of the Wigner's friend-type scenarios, where agents have to compare their experiences (purely subjective in the case of QBism), becomes unclear.


\bibliographystyle{eptcs}


\end{document}